\documentclass[conference]{IEEEtran}
\IEEEoverridecommandlockouts
\usepackage{cite}
\usepackage{amsmath,amssymb,amsfonts}
\usepackage{amsthm}
\usepackage{algorithmic}
\usepackage{graphicx}
\usepackage{textcomp}
\usepackage{xcolor}
\usepackage{booktabs}
\usepackage{xspace}
\usepackage{multirow}
\usepackage{soul}
\usepackage{wrapfig}
\ifCLASSOPTIONcompsoc
\usepackage[caption=false,font=normalsize,labelfont=sf,textfont=sf]{subfig}
\else
\usepackage[caption=false,font=footnotesize]{subfig}
\fi
\newcommand{\ie}{\emph{i.e.},\xspace}
\newcommand{\eg}{\emph{e.g.},\xspace}
\newcommand{\pp}[1]{\left( #1 \right)}

\newtheorem{theorem}{Theorem}[section]
\newtheorem{proposition}[theorem]{Proposition}

\def\BibTeX{{\rm B\kern-.05em{\sc i\kern-.025em b}\kern-.08em
    T\kern-.1667em\lower.7ex\hbox{E}\kern-.125emX}}

\begin{document}

\title{Are ID Embeddings Necessary? Whitening Pre-trained Text Embeddings for Effective Sequential Recommendation
}

\author{\IEEEauthorblockN{Lingzi Zhang\IEEEauthorrefmark{2}\IEEEauthorrefmark{3}, Xin Zhou\IEEEauthorrefmark{3}, Zhiwei Zeng\IEEEauthorrefmark{2}, and Zhiqi Shen\IEEEauthorrefmark{2}\IEEEauthorrefmark{3}}
\IEEEauthorblockA{
\IEEEauthorrefmark{2}School of Computer Science and Engineering, Nanyang Technological University, Singapore\\
\IEEEauthorrefmark{3}Alibaba-NTU Singapore Joint Research Institute, Nanyang Technological University, Singapore\\
Email: lingzi001@e.ntu.edu.sg, xin.zhou@ntu.edu.sg, zhiwei.zeng@ntu.edu.sg, zqshen@ntu.edu.sg}
\vspace{-1em}

}


\maketitle

\begin{abstract}
Recent sequential recommendation models have combined pre-trained text embeddings of items with item ID embeddings to achieve superior recommendation performance. Despite their effectiveness, the expressive power of text features in these models remains largely unexplored. While most existing models emphasize the importance of ID embeddings in recommendations, our study takes a step further by studying sequential recommendation models that only rely on text features and do not necessitate ID embeddings.
Upon examining pre-trained text embeddings experimentally, we discover that they reside in an anisotropic semantic space, with an average cosine similarity of over 0.8 between items.
We also demonstrate that this anisotropic nature hinders recommendation models from effectively differentiating between item representations and leads to degenerated performance. To address this issue, we propose to employ a pre-processing step known as whitening transformation, which transforms the anisotropic text feature distribution into an isotropic Gaussian distribution. 
Our experiments show that whitening pre-trained text embeddings in the sequential model can significantly improve recommendation performance. However, the full whitening operation might break the potential manifold of items with similar text semantics. 
To preserve the original semantics while benefiting from the isotropy of the whitened text features, we introduce WhitenRec+, an ensemble approach that leverages both fully whitened and relaxed whitened item representations for effective recommendations. 
We further discuss and analyze the benefits of our design through experiments and proofs.
Experimental results on three public benchmark datasets demonstrate that WhitenRec+ outperforms state-of-the-art methods for sequential recommendation. 
\end{abstract}


\section{Introduction}
\label{sec:intro}

The sequential recommendation is a subfield of recommender systems that aims to provide personalized item recommendations to users over time. It considers the order in which items are consumed by users to predict the next item the user is likely to interact with~\cite{kang2018self,ma2019hierarchical,xie2022contrastive,zhou2020s3,zhang2019feature,hou2022towards,sun2019bert4rec}. 
Recently, there has been an upsurge of interest in developing sequential recommendation methods that integrate textual information about items, such as product attributes~\cite{zhou2020s3,zhang2019feature,xie2022decoupled}, descriptions~\cite{hou2022towards,zhou2022bootstrap,wei2020graph}, and reviews~\cite{shuai2022review,tran2022aligning}, with ID embeddings to generate more accurate and relevant recommendations.
These recommendation frameworks usually align text features with ID embeddings, highlighting the significance of ID embeddings in recommendations.  
However, there is a conspicuous absence of research exploring sequential recommendation based solely on text features. 
In this paper, we refer to recommendation methods that only use item text features as \textit{text-based recommendation models}. 
We argue that studying text-based recommendation models, which do not necessitate ID embeddings, offers three primary advantages.
Firstly, these models can greatly improve the performance in cold-start scenarios.
E-commerce platforms introduce thousands of new products daily, and conventional sequential models typically integrate random ID embeddings with pre-trained text embeddings to provide recommendations for these new products. 
\textcolor{black}{The integration of random initialized ID embeddings for these new items may inevitably have a detrimental effect on the performance of recommender systems.}
Secondly, text-based recommendation models can be more efficient than those that require ID embeddings. 
Using only text embeddings simplifies demands for both tensor storage and computational resource, as there is no requirement to maintain a large and frequently updated ID embedding matrix.
Lastly, text embeddings are transferable across platforms or domains, whereas ID embeddings are not since user IDs and item IDs are typically not shared in practice.

However, effectively implementing sequential recommendations with only pre-trained text embeddings is non-trivial. Most existing sequential recommendation models~\cite{hou2022towards,yuan2023go,hou2023learning} directly utilize text embeddings extracted by pre-trained language models (\eg BERT~\cite{vaswani2017attention}). 
\textcolor{black}{One recent work~\cite{yuan2023go} posits that superior performance from text-based recommendation models is achieved exclusively with advanced pre-trained language models. We contend that these models do not optimally harness pre-trained text embeddings for sequential recommendation.} To identify potential issues with these pre-trained text embeddings, we first examine their cosine similarities on three recommendation datasets. Our analysis reveals that the pre-trained text embeddings exhibit a notably high average cosine similarity of approximately 0.8, indicating that their embedding spaces are highly anisotropic. 
We then conduct a quantitative analysis to assess the impact of embedding anisotropy on recommendation performance by comparing the performance of an ID-based method and a text-based method adapted from a widely used framework SASRec~\cite{kang2018self}. 
Our results show that the text-based method often yields sub-optimal results compared to the ID-based method. Although the text-based method learns from additional content information, text embeddings appear to be less expressive than standard item ID embeddings and are insufficient to achieve optimal recommendation performance on their own.

To resolve the problem of anisotropy in pre-trained text embeddings, we employ a pre-processing step known as whitening transformation~\cite{kessy2018optimal}, which transforms the pre-trained text embedding distribution into a smooth and isotropic Gaussian distribution and removes the correlation among axes. We name the sequential recommendation model with whitening transformation as {WhitenRec}. Since the primary learning objective for recommendation is to optimize the alignment and uniformity between item representations and sequence representations~\cite{qiu2022contrastive,wang2022towards}, the improved uniformity of sequence representations resulting from whitening transformation leads to enhanced recommendation performance. Notably, {WhitenRec} significantly improves the performance of the sequential recommendation models while using only text features, outperforming the models using ID embeddings, text embeddings, or both embeddings without whitening.
WhitenRec leverages fully whitened representations, where all dimensions are decorrelated and embeddings are uniformly projected into a spherical distribution. 
Although whitening is effective in recommendation, excessive whitening may have a negative impact on the manifold of items that share similar textual semantics.
Therefore, we can also relax the whitening criteria where partial dimensions are decorrleated and the obtained representations tend to preserve more original text semantics at the expense of embedding uniformity~\cite{weng2022an}.
Although the retention of text semantics may appear advantageous for the recommendation task, our experimental results suggest that full whitening leads to the best performance compared to different degrees of relaxed whitening.

To reap the benefits of full whitening while preserving partial semantics in original text features, we propose an ensemble framework {WhitenRec+}, which combines both fully whitened representations and relaxed whitened representations together to enhance item representation learning for the sequential recommendation. Specifically, fully whitened representations are produced by whitening the pre-trained text embeddings with the most stringent whitening to decorrelate across all dimensions. Relaxed whitened representations are produced with less stringent whitening to decorrelate dimensions within each group of dimensions, \ie correlation among groups is kept. The fully whitened item representations and relaxed whitened item representations are subsequently combined by passing them through a shared projection head and summing their outputs. The obtained representations are then processed by the Transformer for sequential recommendation. 
\textcolor{black}{In order to elucidate the efficacy of WhitenRec and WhitenRec+, we undertake both empirical and theoretical investigations, examining representation uniformity and alignment, conditioning, as well as information reconstruction. 
Our initial findings illustrate that both methods can augment the uniformity of user representations, thereby enhancing recommendation performance.  
Subsequently, the conditioning of the transformed item embedding matrix sees improvement in both methods, thus bolstering training stability and optimization. 
Finally, a mathematical analysis reveals that WhitenRec+ outperforms WhitenRec in information preservation, requiring less data for the reconstruction of training inputs.}

In summary, our contributions are the following:
\begin{itemize}
    \item We streamline the existing sequential recommendation framework by studying models that only utilize item text features without the need for ID embeddings.
    Our empirical analysis reveals that anisotropy in pre-trained text embeddings restricts the performance of text-based sequential recommendation models. To resolve this issue, we employ whitening transformation to transform pre-trained text embedding distribution into an isotropic form, which can significantly improve the performance of text-based sequential recommendation models.
    \item Our empirical analysis of the whitening process reveals that it may hurt the manifold of items exhibiting similar textual semantics. To this end, we propose an ensemble approach, {WhitenRec+}, which leverages different degrees of whitening transformations to reap the benefits of full whitening while preserving some of the inherent semantics in the original text features. 
    We conduct a thorough analysis and discussion of the merits of this design in terms of representation uniformity, conditioning, and information reconstruction.
    \item Extensive experiments are conducted on three benchmark datasets to evaluate the performance of the proposed WhitenRec and WhitenRec+ models for the sequential recommendation. Notably, WhitenRec+ outperforms state-of-the-art models across all metrics for all three datasets. 
\end{itemize}

\section{Related Work}
\subsection{Sequential Recommender Systems}
The sequential recommendation problem has gained significant attention in the research community. One of the earliest approaches is Markov Chain-based models~\cite{rendle2010factorizing}, which models the probability of transitions between items in a sequence. However, these models often suffer from the cold-start problem and have limited capability of handling complex sequence patterns. 
Another approach views the sequence as an image and has led to the development of a line of works based on Convolution Neural Network (CNN)~\cite{tang2018personalized,yuan2019simple}. 
Recurrent Neural Network (RNN)~\cite{donkers2017sequential,peng2021ham} has shown remarkable performance in utilizing sequential information for recommendation. 
For example, GRU4Rec~\cite{donkers2017sequential} treats users’ behavior sequences as time series data and uses a multi-layer GRU structure to capture the sequential patterns.
Graph Neural Networks (GNN)~\cite{chang2021sequential,IJCAI-GCL4SR} have been explored to model complex item transition patterns. 
For example, GCL4SR~\cite{IJCAI-GCL4SR} employs a global transition graph and the randomly sampled subgraphs to augment the interaction sequence.
Recently, methods based on transformer architecture~\cite{kang2018self,sun2019bert4rec,zhou2020s3,liu2021augmenting,xie2022contrastive} have shown strong performance in capturing long-range dependencies in a sequence. 
SASRec~\cite{kang2018self} uses a self-attention mechanism and a positional encoding scheme to encode the sequential order of the items. BERT4Rec~\cite{sun2019bert4rec} extends SASRec with a bi-directional self-attention module. CL4SRec~\cite{xie2022contrastive} develops three data augmentation approaches, including item cropping, masking, and reordering, to facilitate contrastive tasks and self-supervised signal extraction.

\subsection{Text-enhanced Recommender Systems}
Recent works~\cite{zhou2020s3,xie2022decoupled,zhang2019feature,hou2022towards,zhou2022bootstrap,zhang2023multimodal,wei2020graph,yuan2023go,zhou2023tale,zhou2023comprehensive} have attempted to leverage textual data of items, such as descriptions, attributes, or brands of products to improve item representations for recommendations. Text-enhanced recommender systems have gained increasing attention due to the explosion of text data and the need for more personalized and informative recommendations.
Some works~\cite{zhou2020s3,xie2022decoupled} focus on modeling item attributes and optimizing the model with attribute prediction task. For example, S$^3$-Rec~\cite{zhou2020s3} adopts a pre-training strategy to predict the correlation between an item and its attributes. In DIF-SR~\cite{xie2022decoupled}, the modeling of item attributes is moved from the input to the attention layer. The attention calculation of auxiliary information and item representation is decoupled to improve the modeling capability of item representations.
With the fast development of Natural Language Processing (NLP) techniques, more works~\cite{zhang2019feature,hou2022towards,zhou2022bootstrap,yuan2023go,zhou2023enhancing} extract the pre-trained features from product descriptions using pre-trained language models. For instance,
in FDSA~\cite{zhang2019feature}, different item features are first aggregated using a vanilla attention layer, followed by a feature-based self-attention block to learn how features transit among items in a sequence. Although these works achieve promising results, they directly utilize the pre-trained text embeddings without analyzing their potential problems. 
\textcolor{black}{
It is noteworthy that UniSRec~\cite{hou2022towards} also proposes utilizing item texts to derive more transferable representations for sequential recommendations. It further involves a linear transformation of the original text representations to mitigate their anisotropy problem. However, our experimental results, as detailed in Sec.~\ref{sec:whiten}, reveal that this parametric approach does not necessarily yield whitened outputs that eliminate correlation across feature dimensions, thereby leading to suboptimal performances.}

\subsection{Whitening}
The whitening, or decorrelation, is a data transformation process with the theoretical guarantee of avoiding collapse by decorrelating each feature dimension~\cite{kessy2018optimal}. 
Among the earliest approaches to whitening is Principal Component Analysis (PCA), first introduced for data analysis and dimensionality reduction~\cite{jolliffe2002principal}, and more recently adapted for use in deep learning~\cite{desjardins2015natural}.
Compared with PCA, Zero-phase Component Analysis (ZCA)~\cite{bell1997independent} whitening introduces an
additional rotation back to the original coordinate system.
Cholesky Decomposition (CD)~\cite{dereniowski2004cholesky} whitening proposed by \cite{siarohin2018whitening} decomposes the covariance matrix into a lower triangular matrix and its conjugate transpose.
Recently, UniSRec~\cite{hou2022towards} adopts a parametric whitening (PW) method which incorporates a linear layer in the whitening transformation for better generalizability.

In the field of deep learning, prior research efforts~\cite{ioffe2015batch,huang2018decorrelated,hua2021feature} explore the application of whitening techniques to the activation of intermediate layers in neural networks. Batch Normalization (BN)~\cite{ioffe2015batch} is the first to perform normalization per mini-batch, thereby enabling back-propagation and reducing the internal covariate shift during training. 
Decorrelated Batch Normalization (DBN)~\cite{huang2018decorrelated} builds upon BN by incorporating ZCA whitening over mini-batch data to further remove correlation among dimensions. 
Lately, another research direction~\cite{ermolov2021whitening,weng2022an,bardes2022vicreg} has emerged, focusing on employing whitening for self-supervised learning, which seeks to avoid the collapse of augmented representations into a single point.
\textcolor{black}{Different from these studies, our work leverages different degrees of decorrelation strength during the whitening process of pre-trained text embeddings to enhance the representation learning for the sequential recommendation.}

\section{\textcolor{black}{Preliminary and Findings}}
\begin{figure*}
  \centering
  \includegraphics[width=0.9\linewidth]{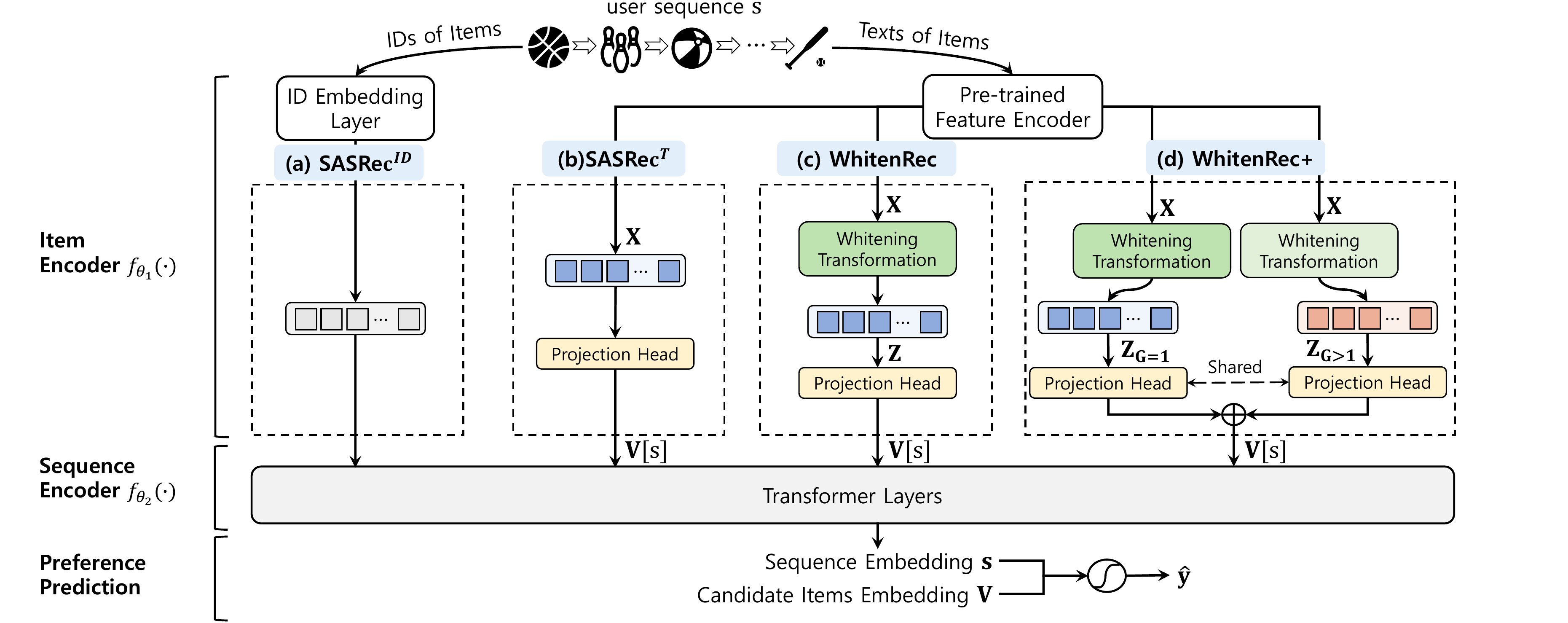}
  \caption{Overall Framework of presenting four variations of sequential recommendation methods, which are SASRec$^{ID}$, SASRec$^{T}$, WhitenRec, and WhitenRec+. Each method is composed of three components, including an item encoder, a sequence encoder, and the preference prediction layer.}
  \label{fig:overall}
  \vspace{-1em}
\end{figure*}

\textcolor{black}{In this section, we begin by outlining an overall framework for sequential recommendation. Subsequently, we delve into our experimental findings concerning anisotropic embeddings, which have been identified to adversely affect the recommendation performance.}

\subsection{Overall Framework}
We present a general framework for sequential recommendation models adapted from SASRec~\cite{kang2018self}, as depicted in Fig.~\ref{fig:overall}.
The framework comprises three major components: an item encoder to extract the latent features of items, a sequence encoder to derive the sequence embedding as the user representation, and a prediction layer to predict the next item. 
Typically, the standard Transformer~\cite{vaswani2017attention} architecture is used as the sequence encoder, and prediction is made based on the inner product between the user \textcolor{black}{(\ie item sequence)} and item representations. 
Fig.~\ref{fig:overall}a illustrates the base model SASRec, which we refer to as SASRec$^{ID}$ throughout the remainder of this paper to facilitate comparison.

The above framework involves a given set of items, denoted as $\mathcal{I}$, an item embedding table $\mathbf{E}\in \mathbb{R}^{d_0\times |\mathcal{I}|}$, and a user sequence $s$ that is comprised of a chronological sequence of items from $\mathcal{I}$. Here, $|\mathcal{I}|$ represents the size of the set and $d_0$ represents the dimension size. 
The recommendation objective is to minimize the cross-entropy loss by training the model parameters $\theta_1$ and $\theta_2$. The objective function $\mathcal{L}$ can be formally represented as:
\begin{gather}
    \mathcal{L} = - \log (\hat{\mathbf{y}}) \, \textrm{ONE-HOT}(\mathbf{y}), 
    \quad
    \hat{\mathbf{y}} = \text{softmax}(\mathbf{V}\mathbf{s}),  \\
    \mathbf{V} = f_{\theta_1}(\mathbf{E}), 
    \quad
    \mathbf{s} = f_{\theta_2}(\mathbf{V}[s]), 
\end{gather}
where $\mathbf{y}$ is the ground-truth next item given a user sequence $s$. $f_{\theta_1}(\cdot)$ is the item encoder that contains an embedding layer and/or a projection head and $f_{\theta_2}(\cdot)$ is the sequence encoder leveraging the Transformer~\cite{vaswani2017attention}.
$\mathbf{V}\in \mathbb{R}^{d \times |\mathcal{I}|}$ denotes the embedding matrix of all items output from the item encoder.
$\mathbf{V}[s]\in \mathbb{R}^{d \times |s|}$ retrieves embeddings of all items in $s$ from $\mathbf{V}$.
Following prior works~\cite{zhou2020s3,qiu2022contrastive,xie2022contrastive}, the hidden vector corresponding to the last position of the sequence from Transformer is selected to be the user representation $\mathbf{s} \in \mathbb{R}^{d \times 1}$.

\subsection{Anisotropic Embedding Space Induces Poor Recommendation Performance}
\label{sec:problem}

Recent research in the field of NLP has revealed that BERT sentence embeddings tend to degenerate into an anisotropic shape, which is referred to as the \textit{representation degeneration problem}~\cite{li2020emnlp,su2021whitening,ethayarajh-2019-contextual}. The embeddings are pushed into a similar direction that is negatively correlated with most hidden states, thus clustering in a narrow cone region of the embedding space. This phenomenon can result in high semantic similarities among embeddings and limit the effectiveness of sentence embeddings. Moreover, it has been demonstrated that this representation degeneration problem can adversely impact the performance of downstream language modeling tasks as well~\cite{li2020emnlp,su2021whitening}.
Since BERT embeddings are commonly utilized by text-based recommendation models to extract text information of items, we conduct a preliminary investigation to determine whether the representation degeneration problem also affects the performance of text-based sequential recommendation models.

We study three datasets from Amazon~\cite{ni2019justifying}, including Arts, Toys, and Tools. Following~\cite{hou2022towards,zhang2021mining,zhou2022bootstrap}, we first concatenate titles, categories, and brands of items as their text descriptions. Next, for each item, a special learnable symbol $[\texttt{CLS}]$ is prepended to the beginning of its text descriptions, after which the concatenated text sequence is processed by the BERT~\cite{devlin2018bert}. We use the output of $[\texttt{CLS}]$ as the text embedding of the item, which is a 768-dimensional vector.

\begin{figure}
  \begin{center}
    \includegraphics[width=0.25\textwidth]{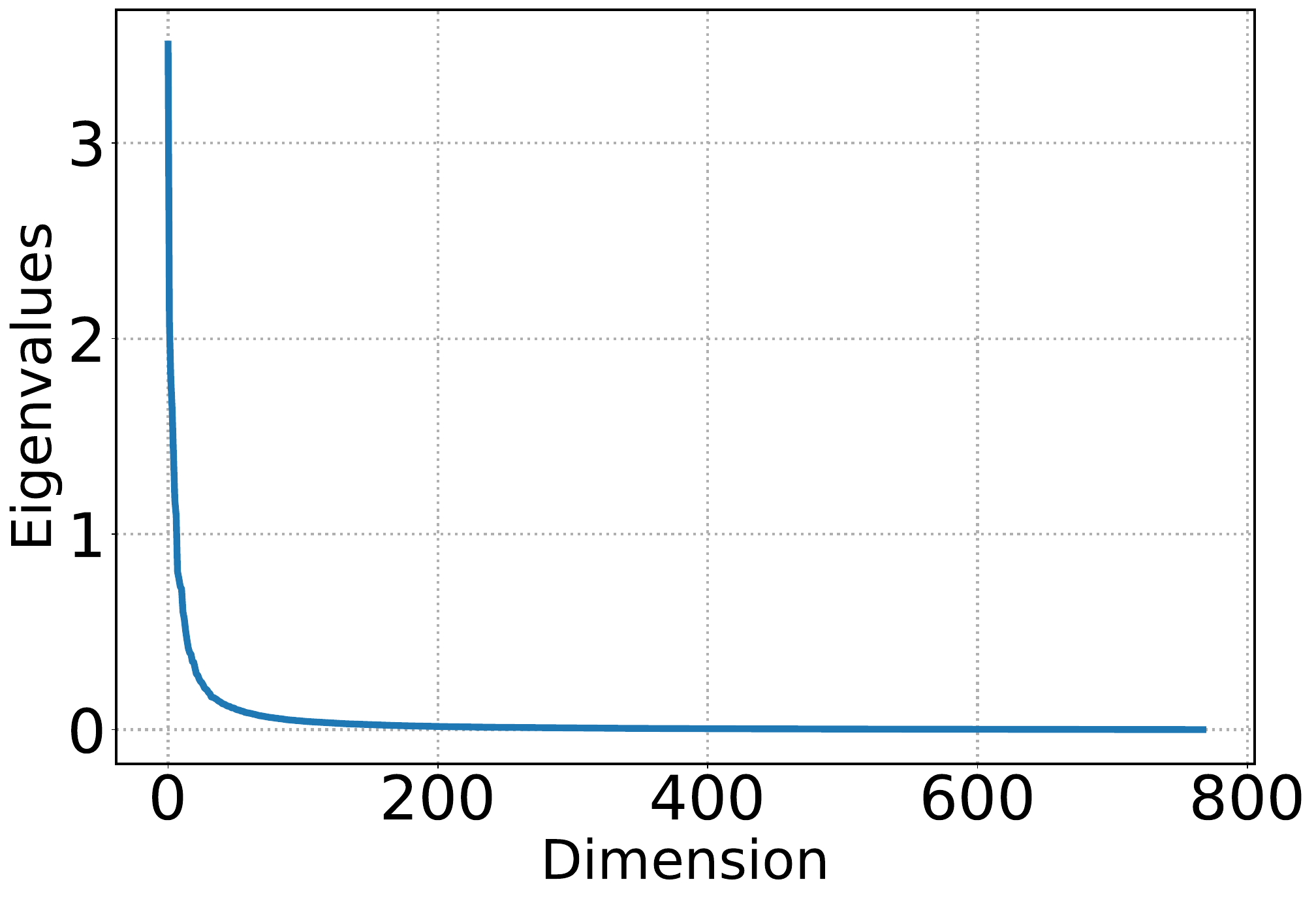}
  \end{center}
  \vspace{-1.5em}
  \caption{Normalized singular values of item text embeddings for Arts. For concision, we omit plots of other datasets as they exhibit similar trends.}
  \label{fig:svd}
  \vspace{-1.5em}
\end{figure}
To show that pre-trained text embeddings in these three datasets also suffer from representation degeneration, we plot their singular values in Fig. ~\ref{fig:svd} and observe a rapid decrease in small values. 
This suggests an anisotropic nature in which one dimension is dominant while the effectiveness of other dimensions is limited.
Additionally, for each item pair (\ie different items) in a dataset, we calculate the cosine similarity based on their pre-trained text embeddings. 
The average cosine similarities of all item pairs for Arts, Toys, and Tools datasets are 0.85, 0.84, and 0.85 respectively.
Indeed, item representations are presented with high cosine similarities, which indicates that their semantic similarities are high and their embedding distributions are highly anisotropic.
Therefore, it is difficult to distinguish between items that use semantically different texts but are close to each other in the embedding space. 
The above analysis demonstrates that the pre-trained text embeddings in recommendation domain also manifest the representation degeneration problem.

\begin{figure*}
     \centering
     \subfloat[Pre-trained text embedding]{\includegraphics[width=0.22\textwidth]{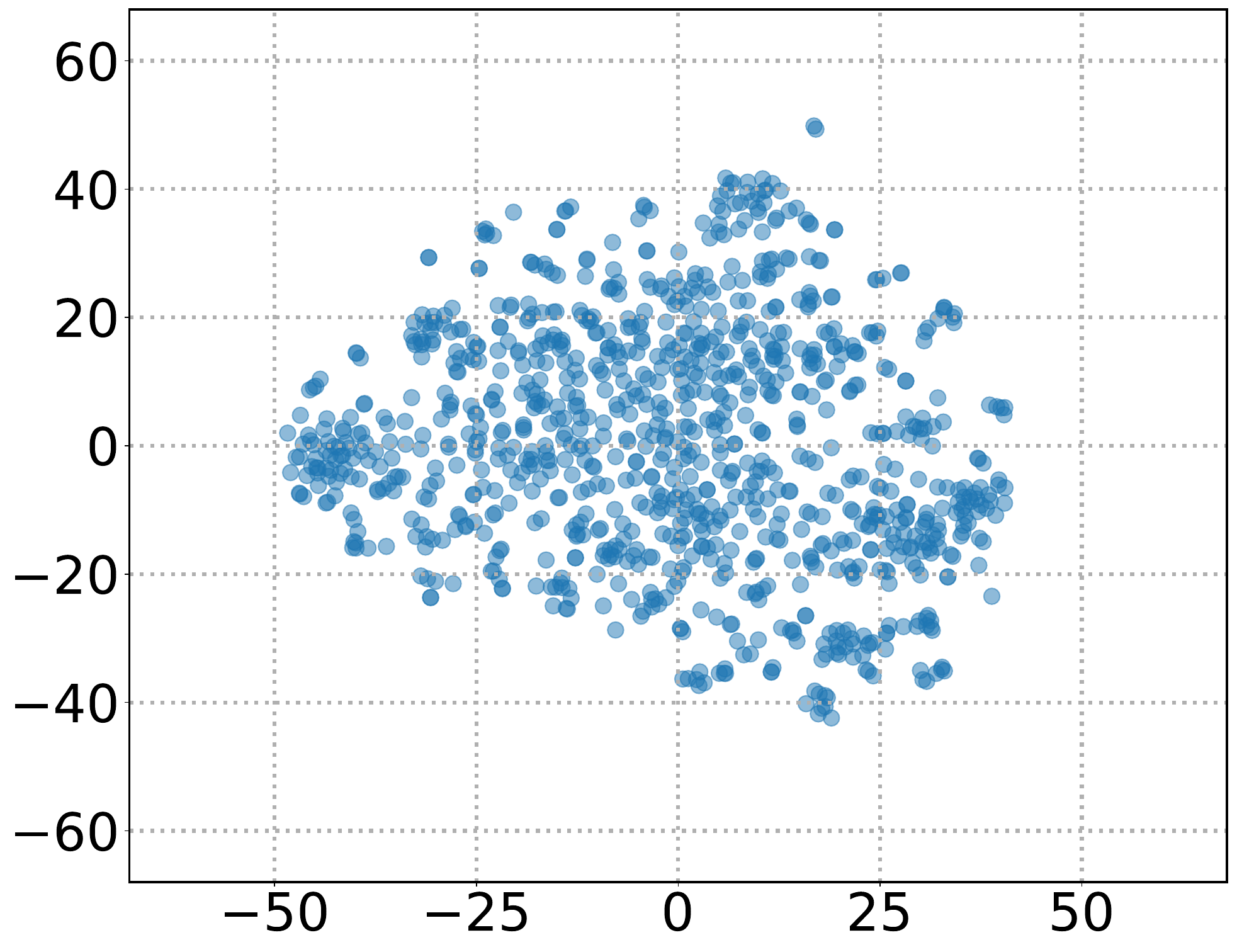}
     \label{fig:tsne-text}}
     \
     \subfloat[Whitened embedding $G=1$]{\includegraphics[width=0.22\textwidth]{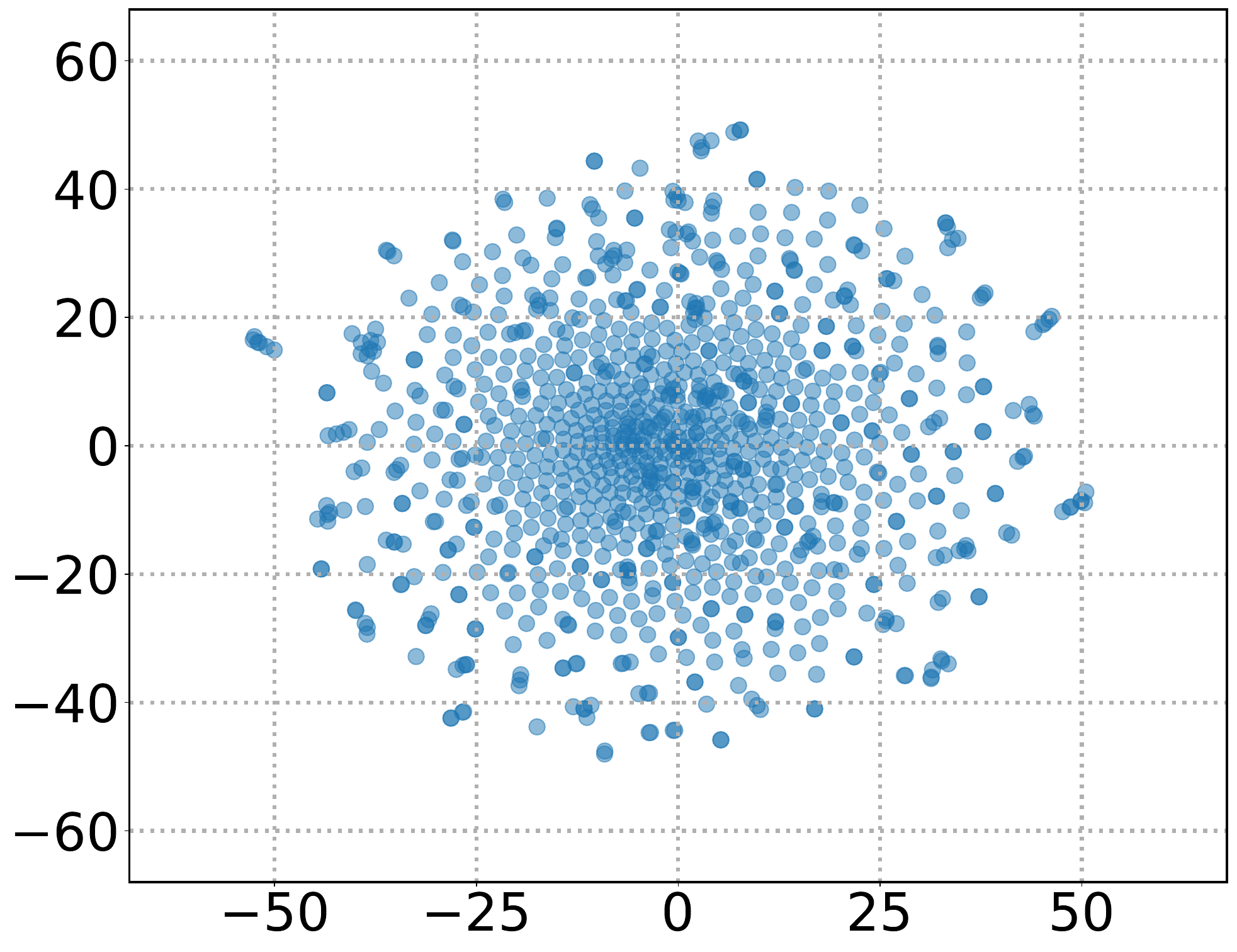}
     \label{fig:tsne-g1}}
     \
     \subfloat[Whitened embedding $G=4$]{\includegraphics[width=0.22\textwidth]{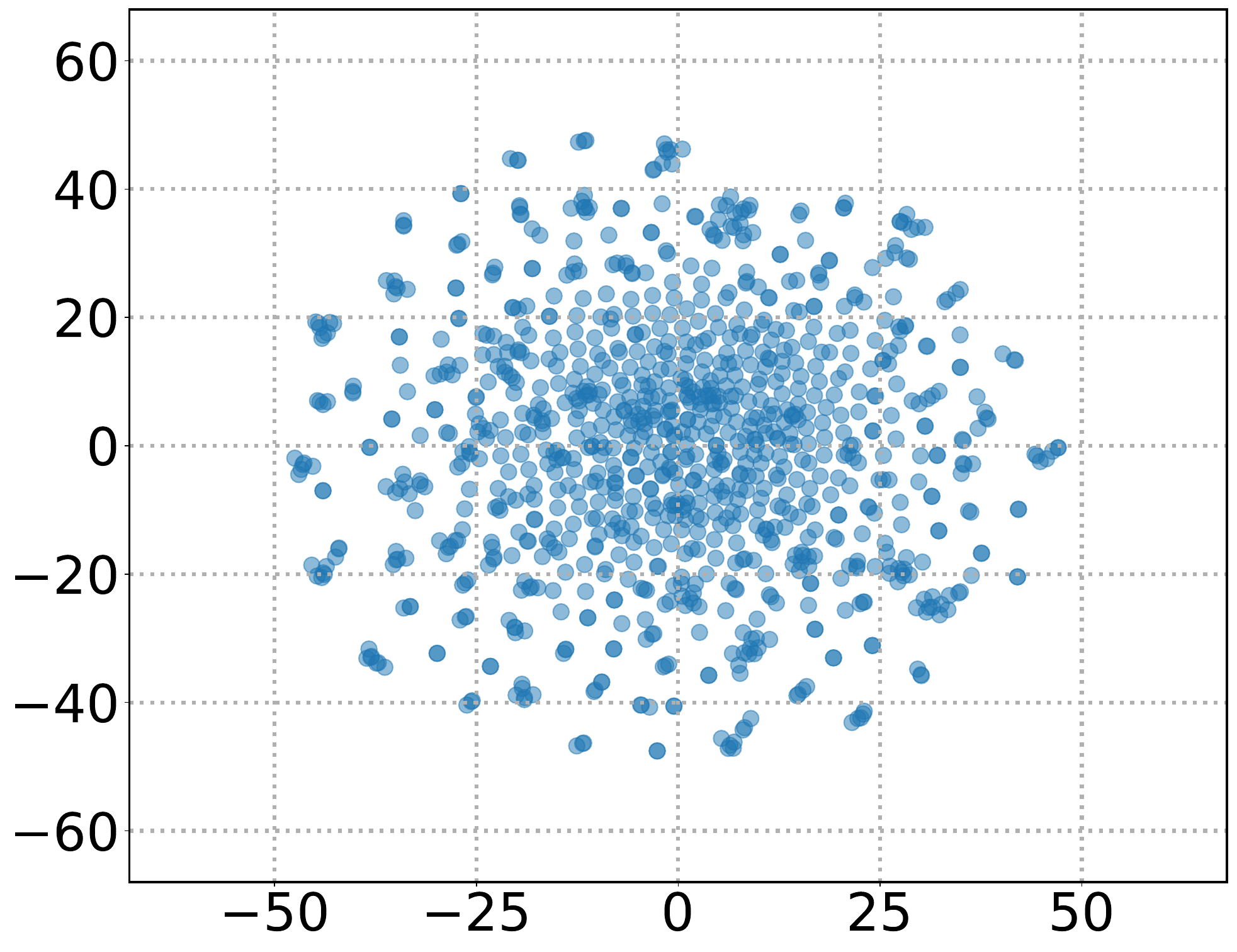}
     \label{fig:tsne-g4}}
     \
     \subfloat[Whitened embedding $G=32$]{\includegraphics[width=0.22\textwidth]{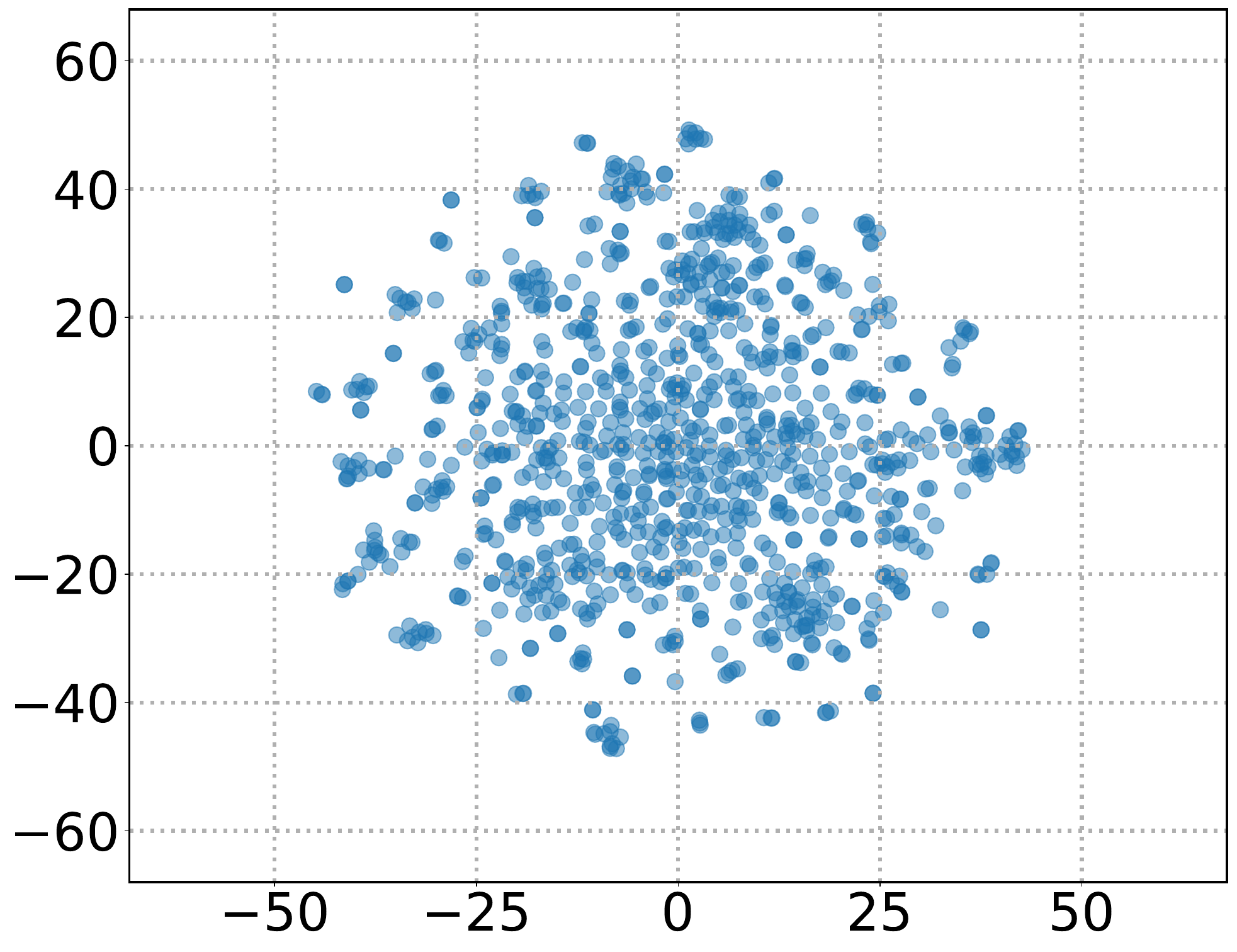}
     \label{fig:tsne-g32}}
     \caption{t-SNE plots of item text embeddings under different settings for Arts.}
     \label{fig:tsne}
     \vspace{-1em}
\end{figure*}

To demystify how the degeneration of representation in pre-trained text embeddings affects recommendation performance, we conduct a quantitative analysis of the independent impact of text embeddings on recommendation performance. 
In particular, we implement a specific instantiation of the general framework (Fig.~\ref{fig:overall}), which we refer to as SASRec$^T$ (Fig.~\ref{fig:overall}b): 
the embedding matrix of items $\mathbf{E}$ is initialized with pre-trained text embedding matrix $\mathbf{X}\in\mathbb{R}^{d_t\times |\mathcal{I}|}$ of items and is not updated during training. $d_t$ is the feature dimension size.
It is worth noting that SASRec$^T$ does not utilize ID embeddings.
$f_{\theta_1}(\cdot)$ is a projector MLP with two hidden layers for feature transformation, where ReLU activation is appended to both hidden layers of the projector. 
We compare the recommendation performance of SASRec$^{T}$ with SASRec$^{ID}$, which does not incorporate text information. 
As presented in Table~\ref{tab:compare}, the application of text embeddings improves Recall, yet it negatively impacts the NDCG in both the Arts and Tools datasets. Considering the Toys dataset, a decline is observed across all evaluated metrics. 
Despite incorporating more informative item contents as opposed to randomly initialized ID embeddings in SASRec$^{ID}$, the effectiveness of SASRec$^T$ is inferior to that of SASRec$^{ID}$ on most datasets.
We suspect that anisotropic item embedding spaces may be the underlying cause of constraining the performance of text-based sequential recommendation models.

\section{\textcolor{black}{Methods: WhitenRec \& WhitenRec+}}
In this section, to address the anisotropy problem, we propose to apply the whitening transformation to eliminate the strong correlation between axes and make text embeddings more isotropic. By doing so, the recommendation performance can be significantly improved. Additionally, we introduce a simple yet effective extension that combines relaxed whitened representations with fully whitened representations, enhancing item representation learning for sequential recommendation.

\subsection{Whitening Transformation to Resolve Anisotropy Problem}
\label{sec:whitening-x}

The anisotropy of pre-trained text embeddings is a widely recognized form of feature degeneration in representation learning. As such, prior works~\cite{hua2021feature, weng2022an} have demonstrated that the application of a whitening transformation~\cite{kessy2018optimal} to project the elements of pre-trained text embeddings onto a spherical distribution can mitigate the anisotropy problem and reduce similarity among distinct instances.
The whitened representation removes the correlation among axes and ensures the item set is scattered in a spherical distribution to avoid the feature collapse with theoretical guarantee~\cite{weng2022an}.

To perform the whitening transformation, given pre-trained text embeddings of all items $\mathbf{X}\in\mathbb{R}^{d_t \times |\mathcal{I}|}$, the whitened output $\mathbf{Z}$ is derived as
\begin{align}
    \mathbf{Z} = \Phi(\mathbf{X}-\mu\cdot\mathbf{1}^\top),
    \label{eq:zca1}
\end{align}
where $\Phi:\mathbb{R}^{d_t \times |\mathcal{I}|}\rightarrow \mathbb{R}^{d_t \times |\mathcal{I}|}$ denotes the function for whitening transformation, $\mu=\frac{1}{|\mathcal{I}|}\mathbf{X}\cdot\mathbf{1}$ is the mean of $\textbf{X}$, $\mathbf{1}$ is the column vector of all ones.
There are many possible ways to perform whitening, including BN~\cite{ioffe2015batch}, PCA~\cite{huang2018decorrelated,kessy2018optimal}, CD~\cite{dereniowski2003cholesky}, and ZCA~\cite{bell1997independent} whitening. Different whitening methods differ in the choice of $\Phi$. 
By default, we choose ZCA whitening, which yields the best performance for most experimental datasets. We also compare different whitening operations and report the details of experimental results in Sec.~\ref{sec:whiten}.
For ZCA whitening, $\Phi$ is defined as follows,
\begin{align}
    \Phi = \mathbf{D}\Lambda^{-\frac{1}{2}}\mathbf{D}^{\top},
    \label{eq:zca2}
\end{align}
where $\Lambda = \text{diag}(\sigma_1,\cdots,\sigma_{d_t})$ and $\mathbf{D}=[\mathbf{d}_1,\cdots,\mathbf{d}_{d_t}]$ are the eigenvalues and associated eigenvectors of $\mathbf{\Sigma} = \mathbf{D}\Lambda\mathbf{D}^\top$. $\mathbf{\Sigma}=\frac{1}{|\mathcal{I}|}(\mathbf{X} - \mu\cdot\mathbf{1}^\top)(\mathbf{X} - \mu\cdot\mathbf{1}^\top)^\top + \epsilon\mathbf{I}$ is the covariance matrix of the centered input $\mathbf{X}$. $\Phi$ ensures the transformed output $\mathbf{Z}$ has the property of $\mathbf{Z}\mathbf{Z}^\top= \mathbf{I}_{d_t}$ to make $\mathbf{X}$ fully whitened. 
Given that the ZCA assumes a full rank covariance matrix, conducting ZCA on all items in which the cardinality of $\mathcal{I}$ significantly exceeds the dimensionality of $d_t$ (\ie $|\mathcal{I}|\gg d_t$) ensures that $\Sigma$ is full rank.
We visualize the t-SNE of item text embeddings before and after ZCA whitening in Fig.~\ref{fig:tsne-text} and Fig.~\ref{fig:tsne-g1}, respectively. We observe that the distribution of item text embeddings that have undergone whitening exhibits spherical symmetry around the origin and is uniformly spread in all directions.

\begin{table}
\caption{\textcolor{black}{Performance comparison of methods without whitening and with whitening. R@20 and N@20 are reported.}}
  \centering
   \small
  \begin{tabular}{l@{\hspace{0.5\tabcolsep}}l@{\hspace{0.5\tabcolsep}}|c@{\hspace{0.5\tabcolsep}}c@{\hspace{0.5\tabcolsep}}|c@{\hspace{0.5\tabcolsep}}c}
    \toprule
    Dataset&Metric& SASRec$^{ID}$ &SASRec$^T$ &WhitenRec & \%Improv\\
    \midrule
    \multirow{2}{*}{Arts}
    &R@20&0.1410&\underline{0.1476}&\textbf{0.1625}& 10.1\%\\
    &N@20&\underline{0.0776}&0.0721&\textbf{0.0796}& 2.6\%\\
    \midrule
    \multirow{2}{*}{Toys}
    &R@20&\underline{0.1121}&0.0983&\textbf{0.1201}& 7.1\% \\
    &N@20&\underline{0.0467}&0.0429&\textbf{0.0521}& 11.6\%\\
    \midrule
    \multirow{2}{*}{Tools}
    &R@20&0.0712&\underline{0.0739}&\textbf{0.0861}& 16.5\%\\
    &N@20&\underline{0.0418}&0.0386&\textbf{0.0453}& 8.4\%\\

    \bottomrule
  \end{tabular}
  \vspace{-1em}
  \label{tab:compare}
\end{table}

We incorporate the ZCA whitening of Eqn.~\eqref{eq:zca1} into SASRec$^{T}$ and refer to the resulting model as WhitenRec, which is illustrated in Fig.~\ref{fig:overall}c. 
\textcolor{black}{As indicated in Table~\ref{tab:compare}, the corresponding representation yields a significant improvement of 10.1\%, 7.1\%, 16.5\% in Recall@20 compared to SASRec$^{ID}$ or SASRec$^{T}$ on Arts, Toys, and Tools respectively.}
It is evident that the application of the whitening transformation in WhitenRec, without introducing any additional trainable parameters, results in a significant improvement in performance.

\begin{figure*}
  \quad
  \centering
  \begin{minipage}[b]{0.179\paperwidth}
    \subfloat[Arts]{\includegraphics[width=\textwidth]{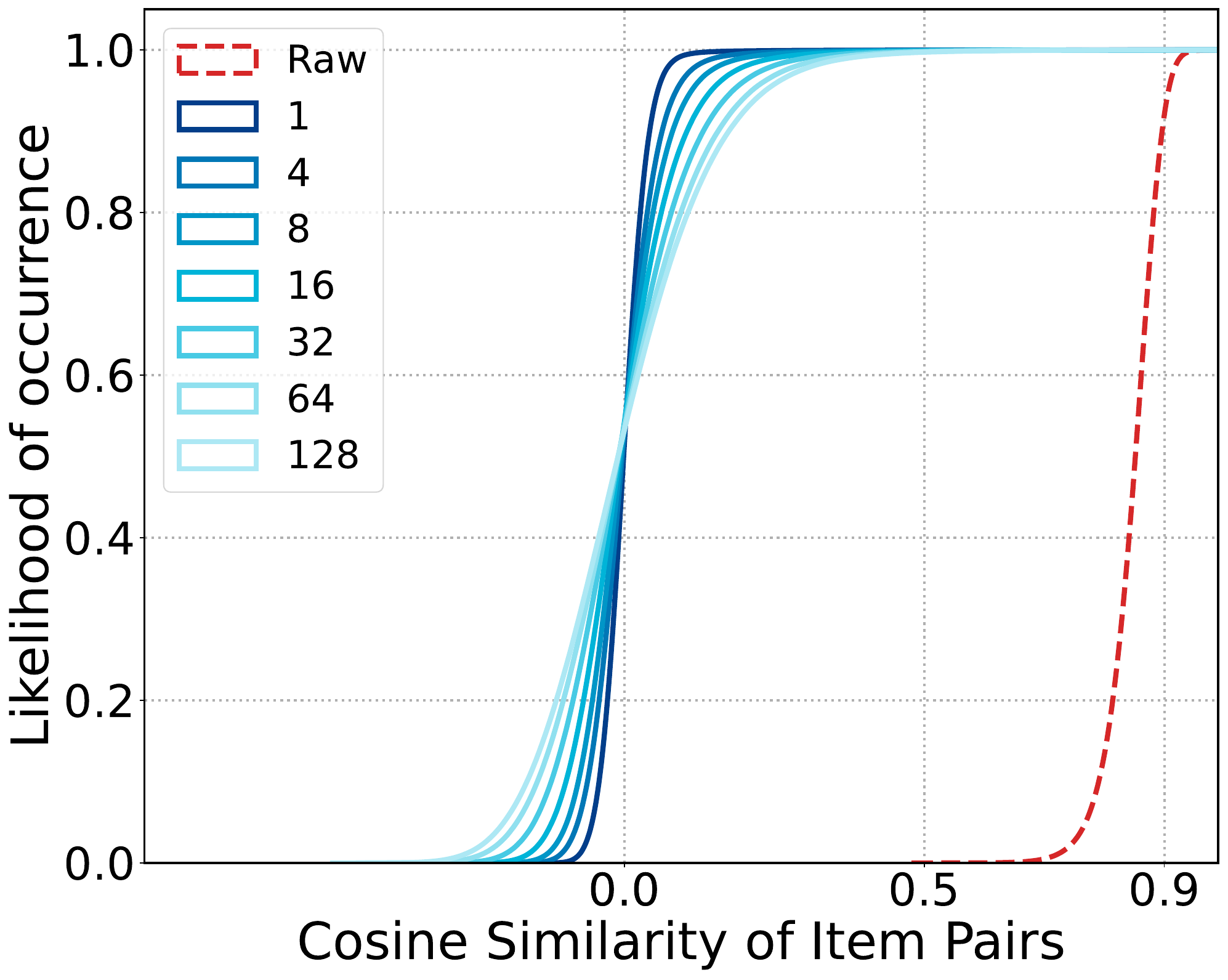}}
    \caption{CDF of item pairs.}
    \label{fig:cdf-items}
  \end{minipage}
  \hfill
  \centering
  \begin{minipage}[b]{0.6\paperwidth}
    \subfloat[Arts]{\includegraphics[width=0.32\textwidth]{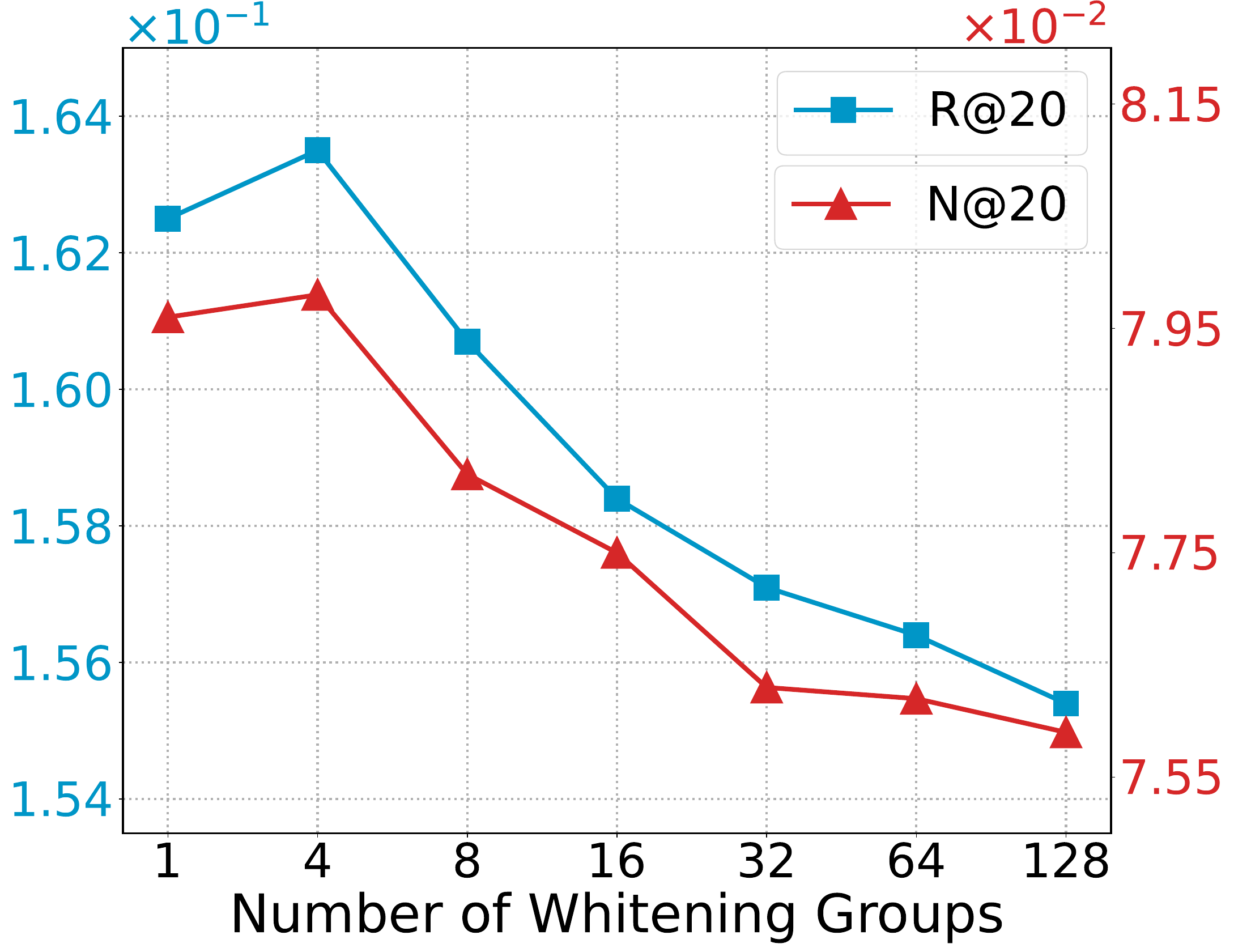}}
    \hfill
    \subfloat[Toys]{\includegraphics[width=0.315\textwidth]{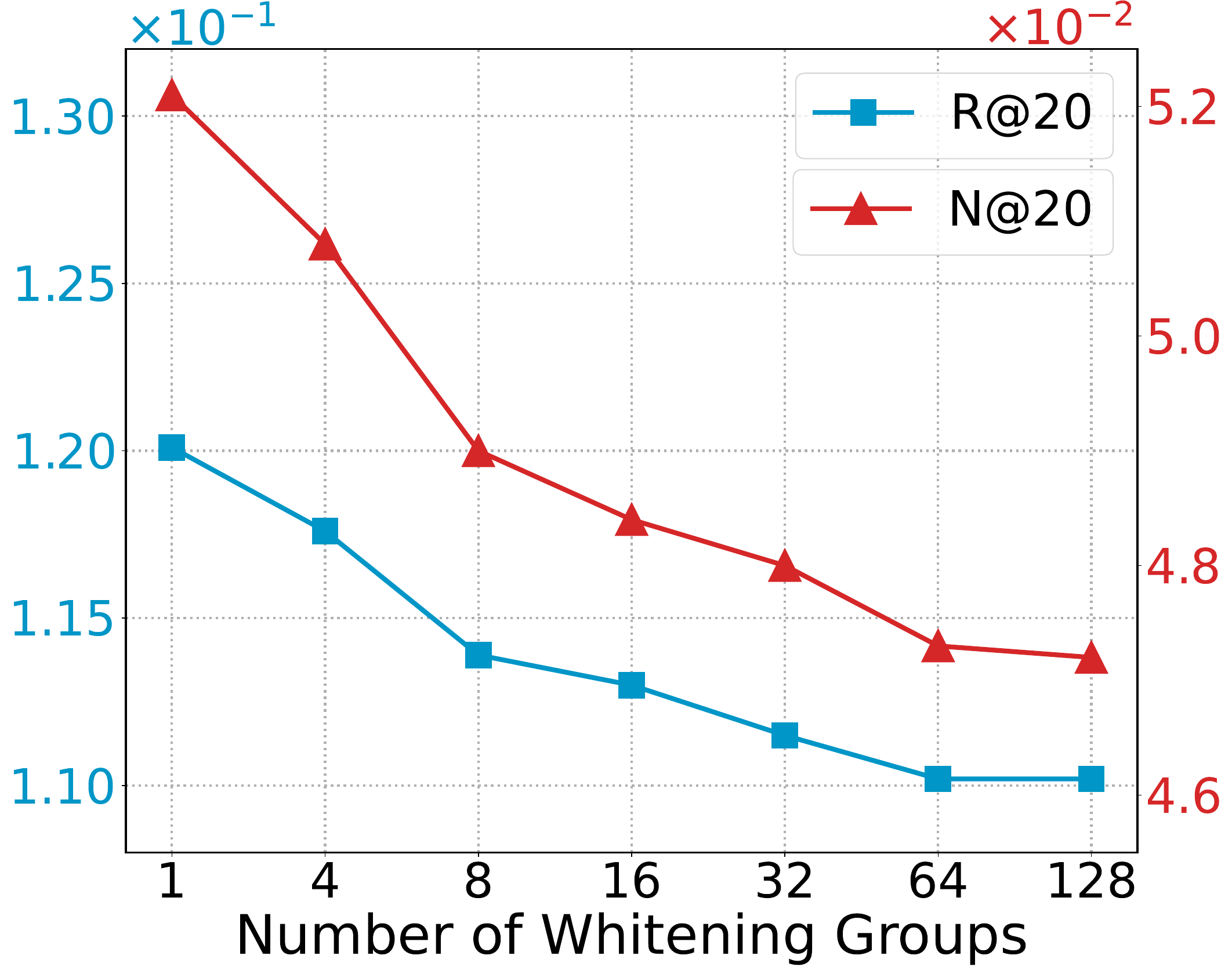}}
    \hfill
    \subfloat[Tools]{\includegraphics[width=0.308\textwidth]{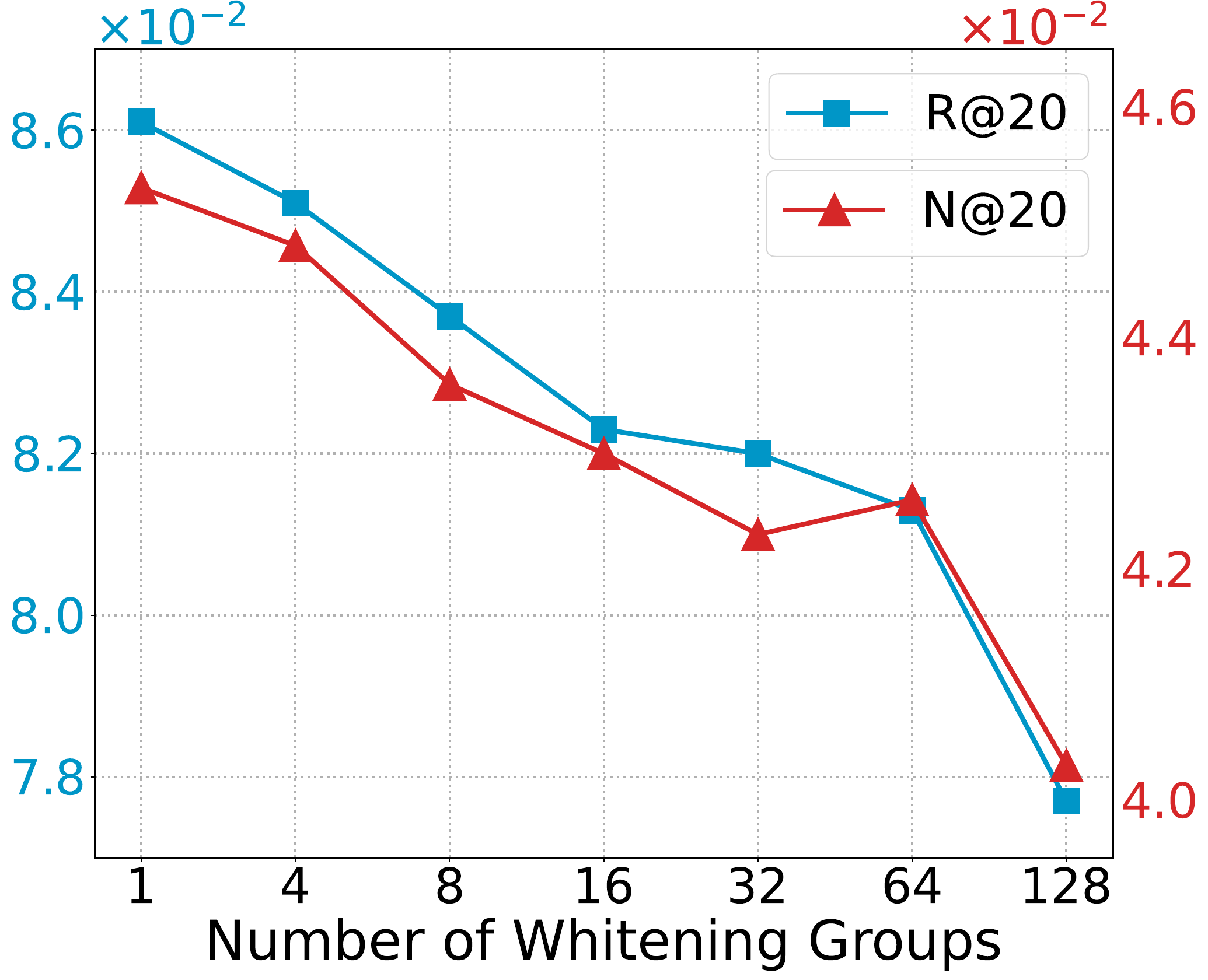}}
    \caption{Performance of different groups $G$ for Whitening.}
    \label{fig:group-whiten}
  \end{minipage}
  \vspace{-1em}
\end{figure*}
\subsection{Relaxed Whitening for Retaining Text Semantics}

Although ZCA whitening can effectively decorrelate pre-trained text embeddings, 
our observations indicate that the fully whitened representation (Fig.~\ref{fig:tsne-g1}) may have an adverse impact on the manifold of items sharing similar textual semantics, in comparison to the original text representation (Fig.~\ref{fig:tsne-text}).

Inspired by~\cite{huang2018decorrelated,hua2021feature}, we adapt ``group whitening'' to standardize covariance matrices within dimensional groups to relax the extent of whitening and retain more original textual semantics. Specifically, the relaxed whitening with the number of groups $G$ takes as its input a matrix $\mathbf{X}\in\mathbb{R}^{d_t \times |\mathcal{I}|}$ and its output is a matrix $\mathbf{Y}\in\mathbb{R}^{d_t \times |\mathcal{I}|}$ computed as:
\begin{align}
    \mathbf{Y}^{[h]} & = \text{ZCA}(\mathbf{X}^{[h]}),
\end{align}
where $\mathbf{X}^{[h]} = \left( \left(\mathbf{X}_{ (h-1) \cdot \frac{d_t}{G} + 1 } \right)^\top, \cdots, \left(\mathbf{X}_{h \cdot \frac{d_t}{G}} \right)^\top \right)^\top \in \mathbb{R}^{\frac{d_t}{G} \times |\mathcal{I}|}$ and $\mathbf{Y}^{[h]} = \left( \left(\mathbf{Y}_{ (h-1) \cdot \frac{d_t}{G} + 1 } \right)^\top, \cdots, \left(\mathbf{Y}_{h \cdot \frac{d_t}{G}} \right)^\top \right)^\top \in \mathbb{R}^{\frac{d_t}{G} \times |\mathcal{I}|}$. $\text{ZCA}(\mathbf{X})=\mathbf{D}\Lambda^{-\frac{1}{2}}\mathbf{D}^{\top}(\mathbf{X}-\mu\cdot\mathbf{1}^\top)$ follows Eqn.~\eqref{eq:zca1} and \eqref{eq:zca2}.
In other words, $\mathbf{Y}$ is derived by dividing all feature dimensions $d_t$ into $G$ groups and applying ZCA whitening to each group independently.

We visualize the Cumulative Distribution Function (CDF) plot of item pairs concerning different extents of whitening on text embeddings of Arts, \ie different $G$, in Fig.~\ref{fig:cdf-items}. The legend specifies $G$ involved in the whitening transformation, with a decreasing extent of whitening resulting from an increase in $G$. ``Raw'' denotes the original text features without whitening. Other datasets showing similar distributions are omitted for space constraints.
Namely, a smaller value of $G$ corresponds to a higher degree of decorrelation, indicating a stronger suppression of redundant information. 
From Fig.\ref{fig:cdf-items}, weaker whitening leads to a less concentrated CDF line within a broader range, indicating increasingly similar item representations and more preserving of textual semantics.

Despite the apparent advantage of retaining text semantics for recommendation tasks, our findings suggest that the exclusive use of relaxed whitened item representations for recommendation may result in cluttered item embedding distributions, as demonstrated in Fig.~\ref{fig:tsne-g1}, Fig.~\ref{fig:tsne-g4} and Fig.~\ref{fig:tsne-g32}.
In these figures, we perform whitening with values of $G$ equal to 1, 4, and 32, respectively. It is apparent that as $G$ increases, the distribution becomes increasingly non-uniform.
To investigate the impact of $G$ on the recommendation performance, we conduct experiments on WhitenRec by varying $G$ in the range of $\{1, 4, 8, 16, 32, 64, 128\}$ and report the results for Arts, Toys, and Tools datasets in Fig.~\ref{fig:group-whiten}. The results indicate that optimal performance is achieved when $G$ is set to a smaller value, suggesting that further decorrelation enhances the representation learning of sequential recommendation.

\subsection{Ensemble of Relaxed Whitening for Further Gains}
The preceding sections present evidence that applying whitening techniques for dimension decorrelation mitigates the feature degeneration issue, consequently enhancing the sequential recommendation performance. Nonetheless, a relaxation of the whitening criteria aiming to retain more of the original text semantics results in sub-optimal performance.

To maximize the advantages of complete whitening while retaining some semantic content from the original text features, we propose an ensemble method WhitenRec+, which leverages both fully whitened representations and relaxed whitened representations to further improve the learning of item representations.
The framework is depicted in Fig.~\ref{fig:overall}d.
Specifically, we apply both the most stringent and relaxed whitening on item text features $\mathbf{X}$. The resultant embeddings are denoted as $\mathbf{Z}_{G=1}$ and $\mathbf{Z}_{G>1}$, respectively. 
Then, we map both $\mathbf{Z}_{G=1}$ and $\mathbf{Z}_{G>1}$ into a latent representation space using the item encoder $f_{\theta_1}(\cdot)$, \ie a shared projection head consisting of two MLP layers. The outputs from $f_{\theta_1(\cdot)}$ are combined using element-wise summation:
\begin{align}
\mathbf{V} = f_{\theta_1}\left(\mathbf{Z}_{G = 1}\right) + f_{\theta_1}\left(\mathbf{Z}_{G>1}\right).
\end{align}
Subsequently, $\mathbf{V}$ is used as the input to the sequence encoder for the generation of recommendations.

Intuitively, the recommendation model can learn representations that are both discriminative and robust to variations in the input data by leveraging both fully whitened representations and relaxed whitened representations. Results in Sec.~\ref{sec:exp} show that WhitenRec+ not only can benefit the warm setting of sequential recommendation but also help further improve the cold setting with a performance improvement of 8.5\%, 17.9\%, and 64.5\% on NDCG@50 for Arts, Toys, and Tools.

\begin{figure*}
     \centering
     \subfloat[Arts-Users]{\includegraphics[width=0.22\textwidth]{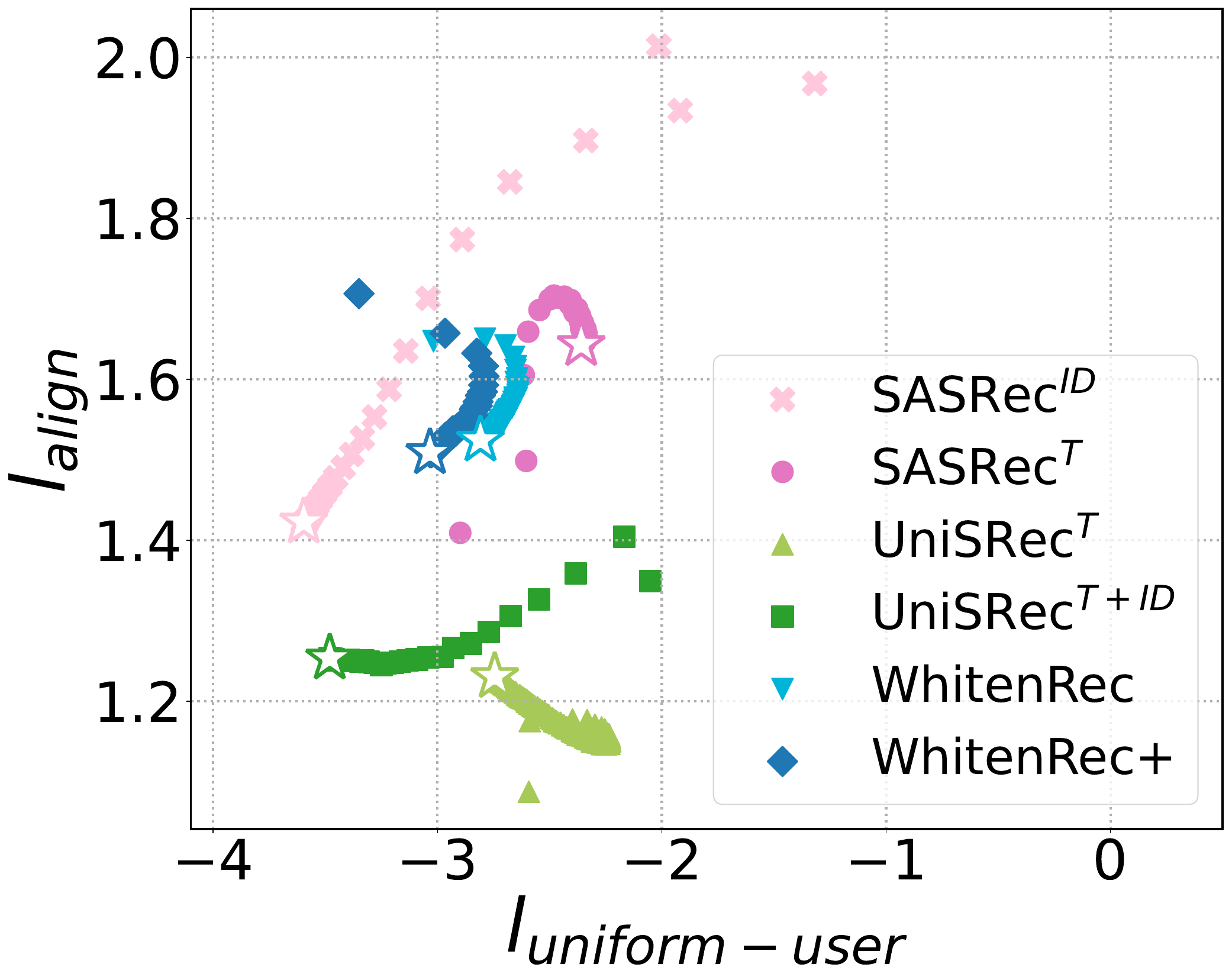}}
     \quad
     \subfloat[Toys-Users]{\includegraphics[width=0.22\textwidth]{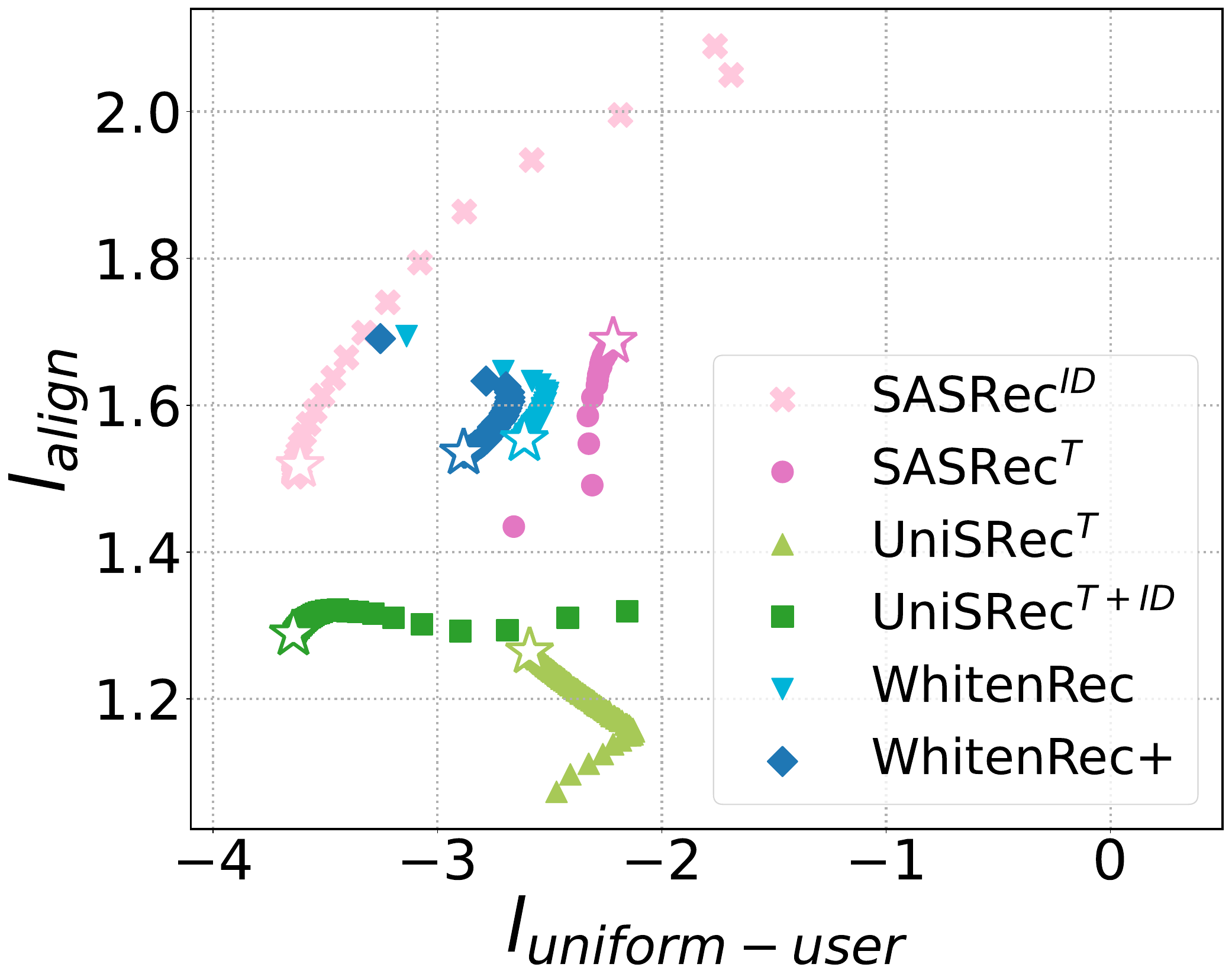}}
     \quad
     \subfloat[Tools-Users]{\includegraphics[width=0.22\textwidth]{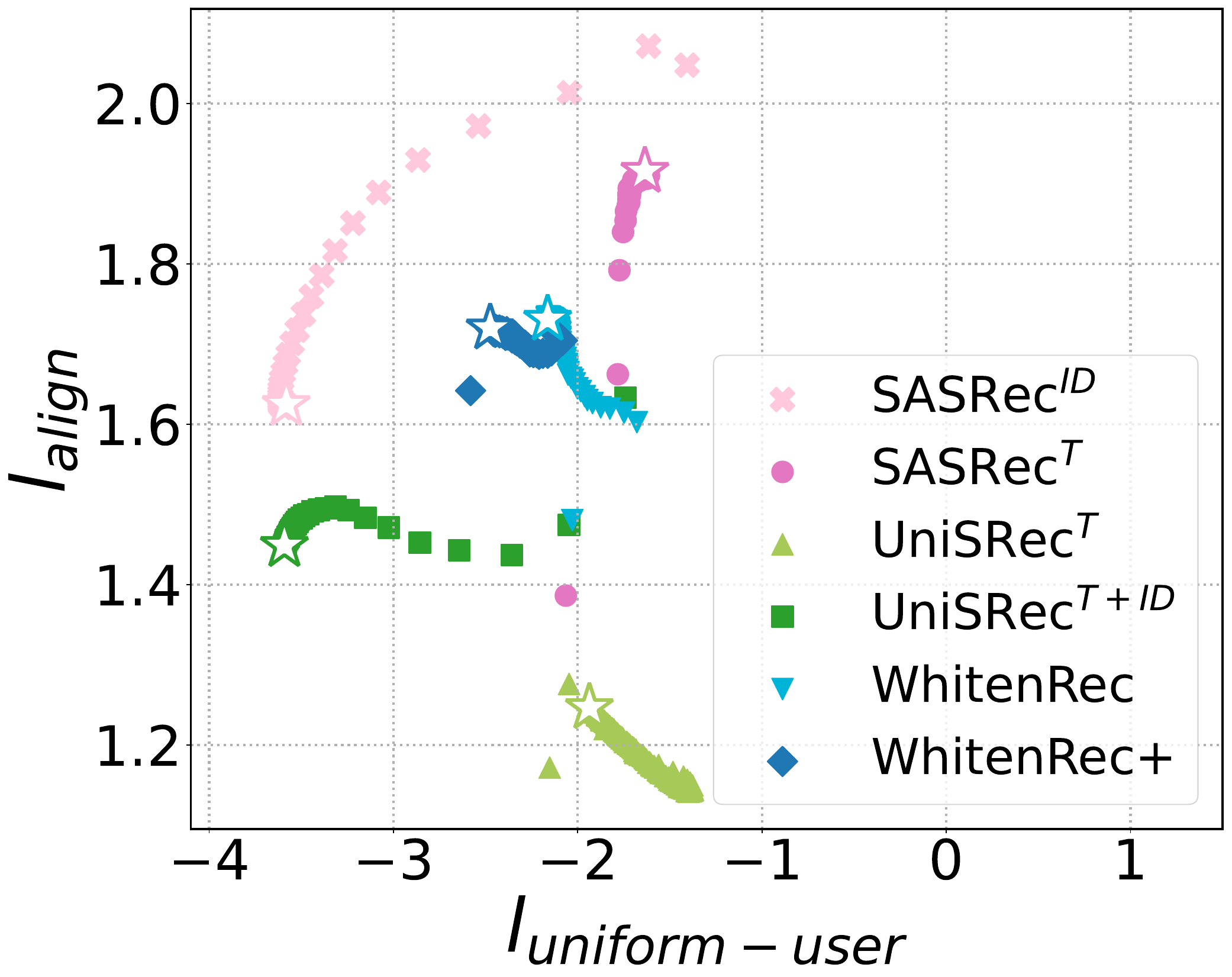}}
     \quad
     \subfloat[Food-Users]{\includegraphics[width=0.22\textwidth]{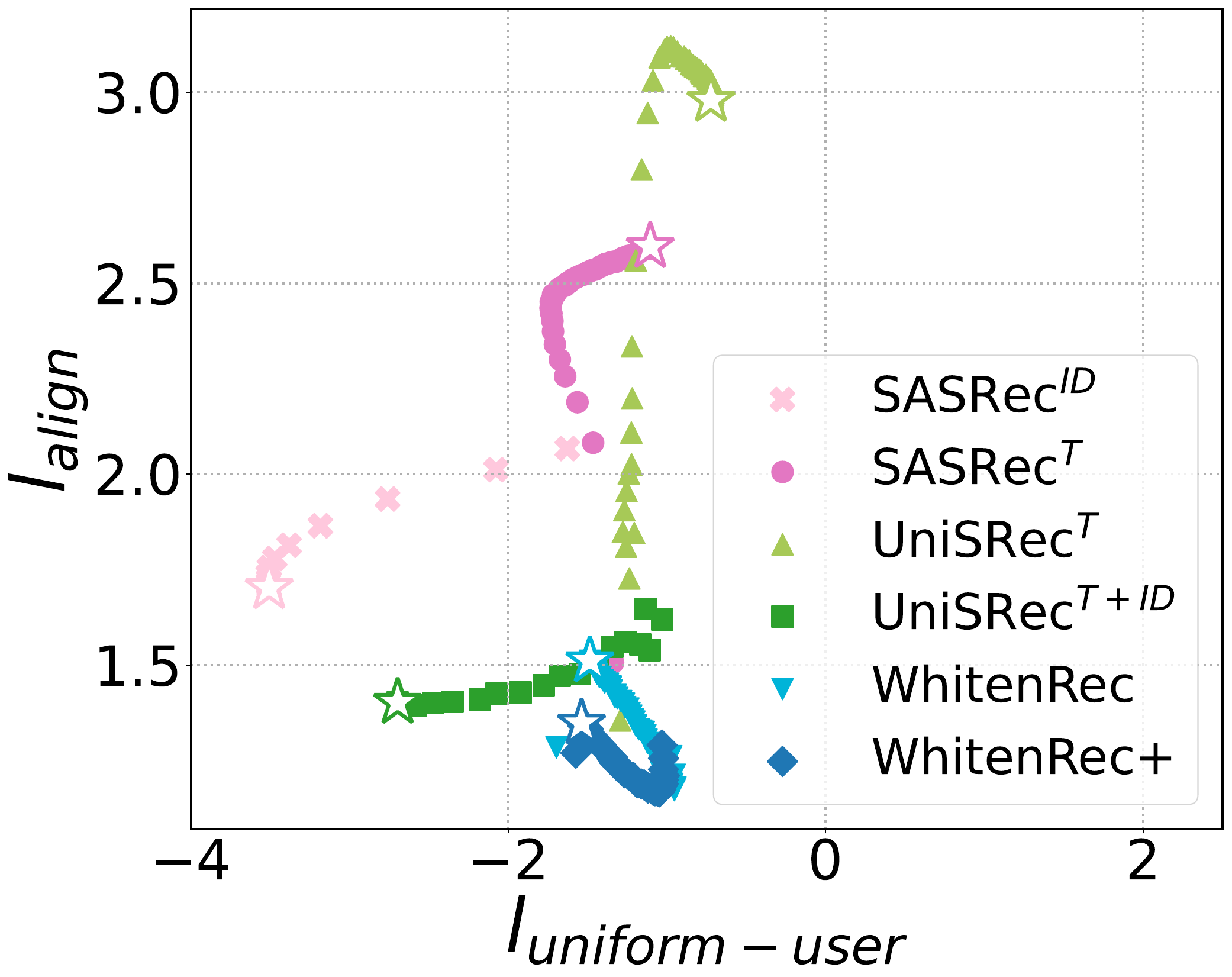}}
     \\
     \subfloat[Arts-Items]{\includegraphics[width=0.22\textwidth]{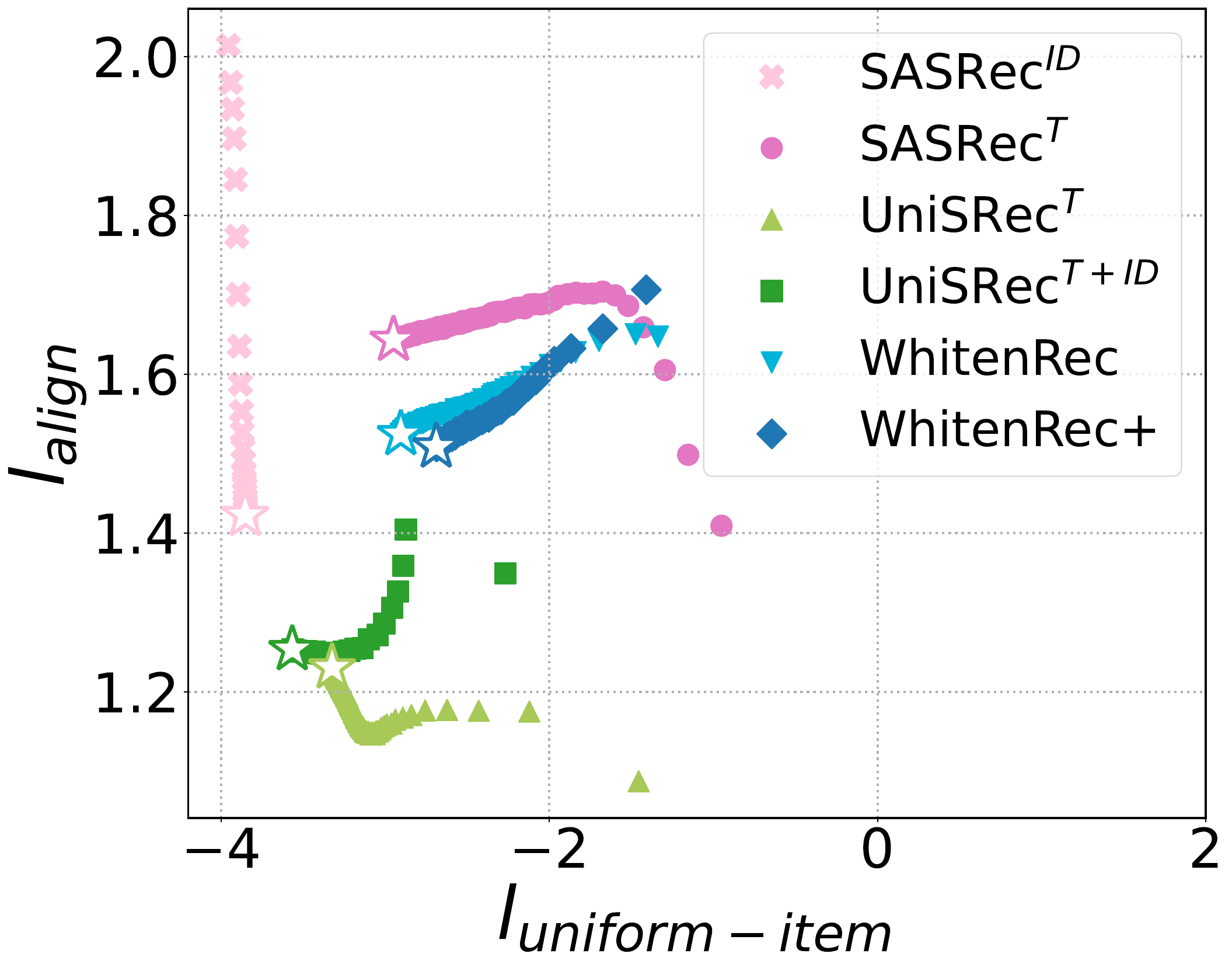}}
     \quad
     \subfloat[Toys-Items]{\includegraphics[width=0.22\textwidth]{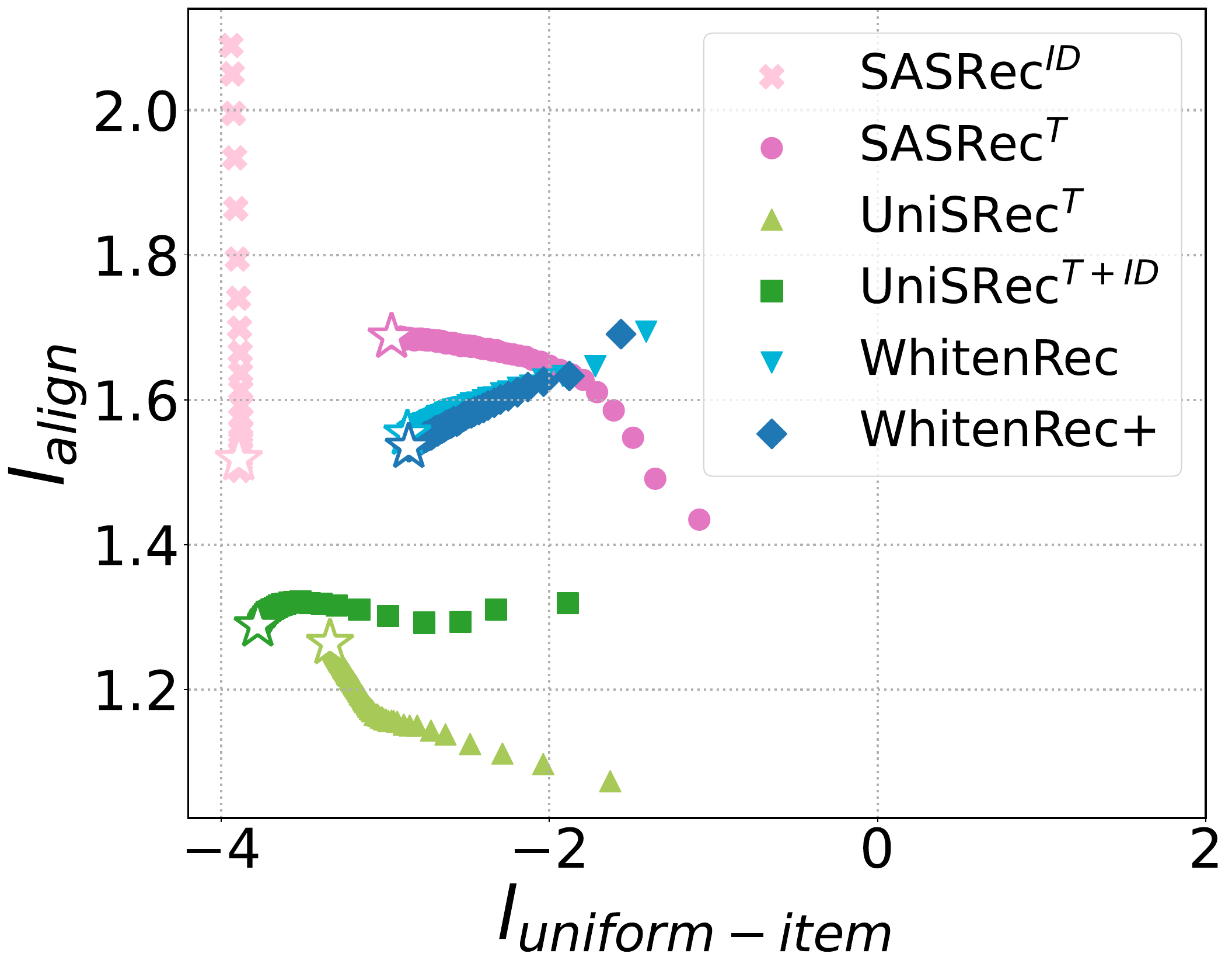}}
    \quad
     \subfloat[Tools-Items]{\includegraphics[width=0.22\textwidth]{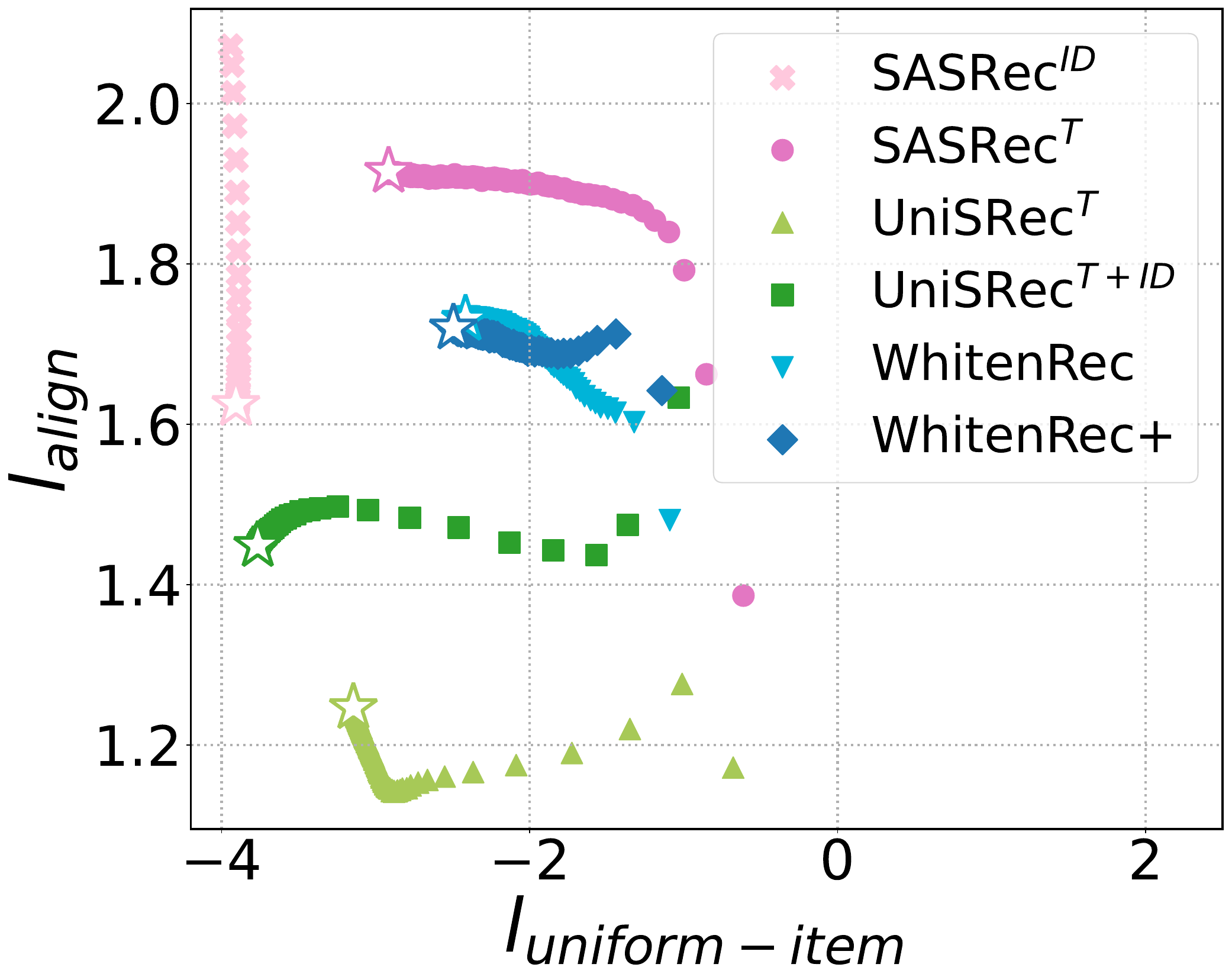}}
     \quad
     \subfloat[Food-Items]{\includegraphics[width=0.22\textwidth]{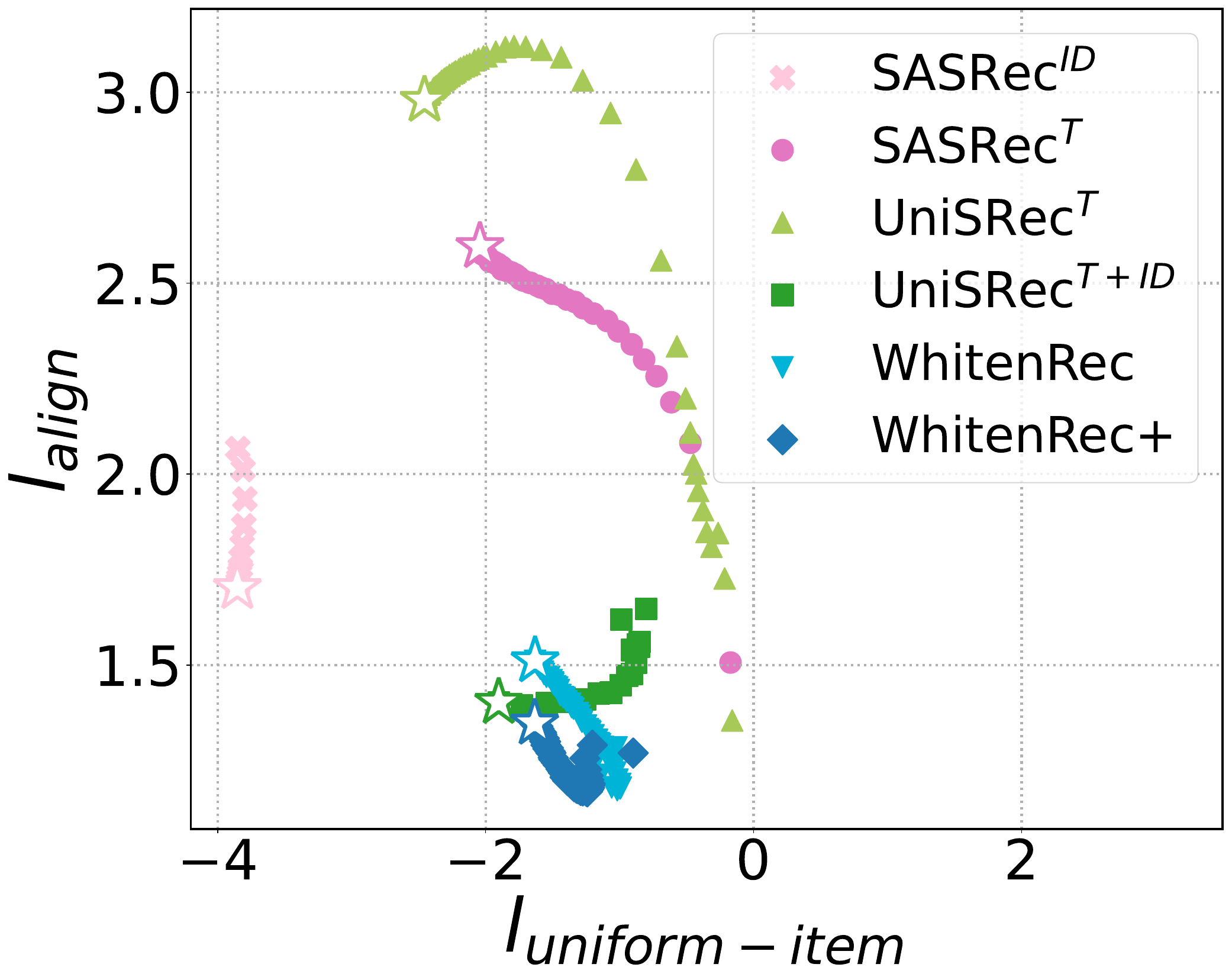}}
     \caption{\textcolor{black}{$l_{align}-l_{uniform}$ plots for representations of users and items during training for Arts, Toys, Tools, and Food datasets. We visualize these two metrics in each epoch, and the stars indicate the last converged epoch. 
     For $l_{align}$ and $l_{uniform}$, lower numbers are generally more preferable. 
     }}
     \label{fig:align_uniform}
     \vspace{-1em}
\end{figure*}

\subsection{Discussion and Analysis}

We further investigate the merits of WhitenRec and WhitenRec+ through both empirical and theoretical analyses. Our examination focuses on representation uniformity and alignment, conditioning, and information reconstruction.

\subsubsection{Uniformity and Alignment}
We analyze the user embedding $\mathbf{s}$ (\ie generated by the sequence encoder) and the item embedding $\mathbf{v}$ retrieved from $\mathbf{V}$ (\ie generated by the item encoder) with respect to their uniformity and alignment~\cite{wang2022towards}. The uniformity and alignment in the context of recommendation are formulated as follows:
\begin{equation}
\begin{aligned}
    l_{align} &= \underset{(u,i)\sim p_{pos}}{\mathbb{E}}|| f(\mathbf{s}_u)-f(\mathbf{v}_i)||^2, \\
    l_{uniform-user} &= log\underset{(u,{u'})\sim p_{user}}{\mathbb{E}}e^{-2|| f(\mathbf{s}_u)-f(\mathbf{s}_{u'})||^2}, \\
    l_{uniform-item} &= log\underset{(i,i')\sim p_{item}}{\mathbb{E}}e^{-2|| f(\mathbf{v}_i)-f(\mathbf{v}_{i'})||^2},
    \label{eq:uni_align}
\end{aligned}
\end{equation}
where $f(\cdot)$ indicates $l_2$ normalized representations. $p_{pos}$, $p_{user}$, and $p_{item}$ are the distribution of positive user-item pairs, users, and items, respectively.

\textcolor{black}{We present visualizations of the learned user and item representations from six models with respect to uniformity and alignment on the Arts, Toys, Tools, and Food datasets in Fig.\ref{fig:align_uniform}. The analysis includes four text-based methods: SASRec$^T$, UniSRec$^T$, WhitenRec, and WhitenRec+, as well as one method utilizing ID embeddings, SASRec$^{ID}$, and another combining text with ID embeddings, UniSRec$^{T+ID}$.
We have two key observations:
(1) When comparing the text-based methods to those incorporating ID embeddings (\ie SASRec$^{ID}$ and UniSRec$^{T+ID}$), it is noted that while the latter exhibit lower uniformity, their performances are worse than our methods. This suggests that the positive correlation between uniformity and performance is limited in scope. Excessive pursuit of uniformity may overlook the proximity of semantically similar items and ultimately impair recommendation performance.
(2) In our comparison of four text-based methods, it is observed that, despite showing higher or comparable item uniformity and alignment relative to SASRec$^T$ and UniSRec$^T$ across most datasets, WhitenRec and WhitenRec+ consistently achieve superior user uniformity.
In conjunction with the performance reported in Table~\ref{tab:overall}, our results indicate that user uniformity plays a more significant role in determining recommendation performance.}

\subsubsection{Conditioning Analysis}
\begin{figure*}
     \centering
     \subfloat[Arts-Condition Number]{\includegraphics[width=0.22\textwidth]{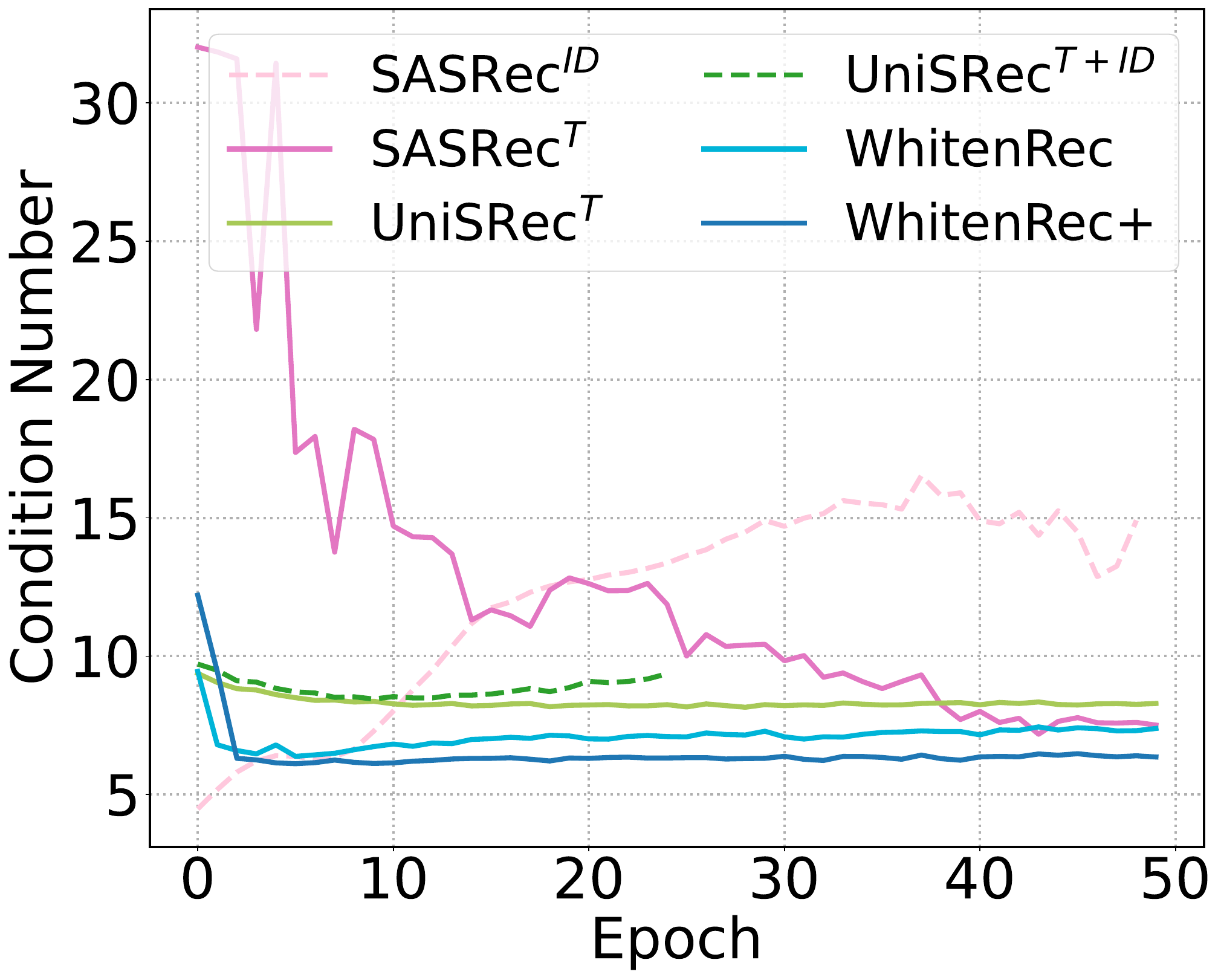}}
     \quad
     \subfloat[Toys-Condition Number]{\includegraphics[width=0.22\textwidth]{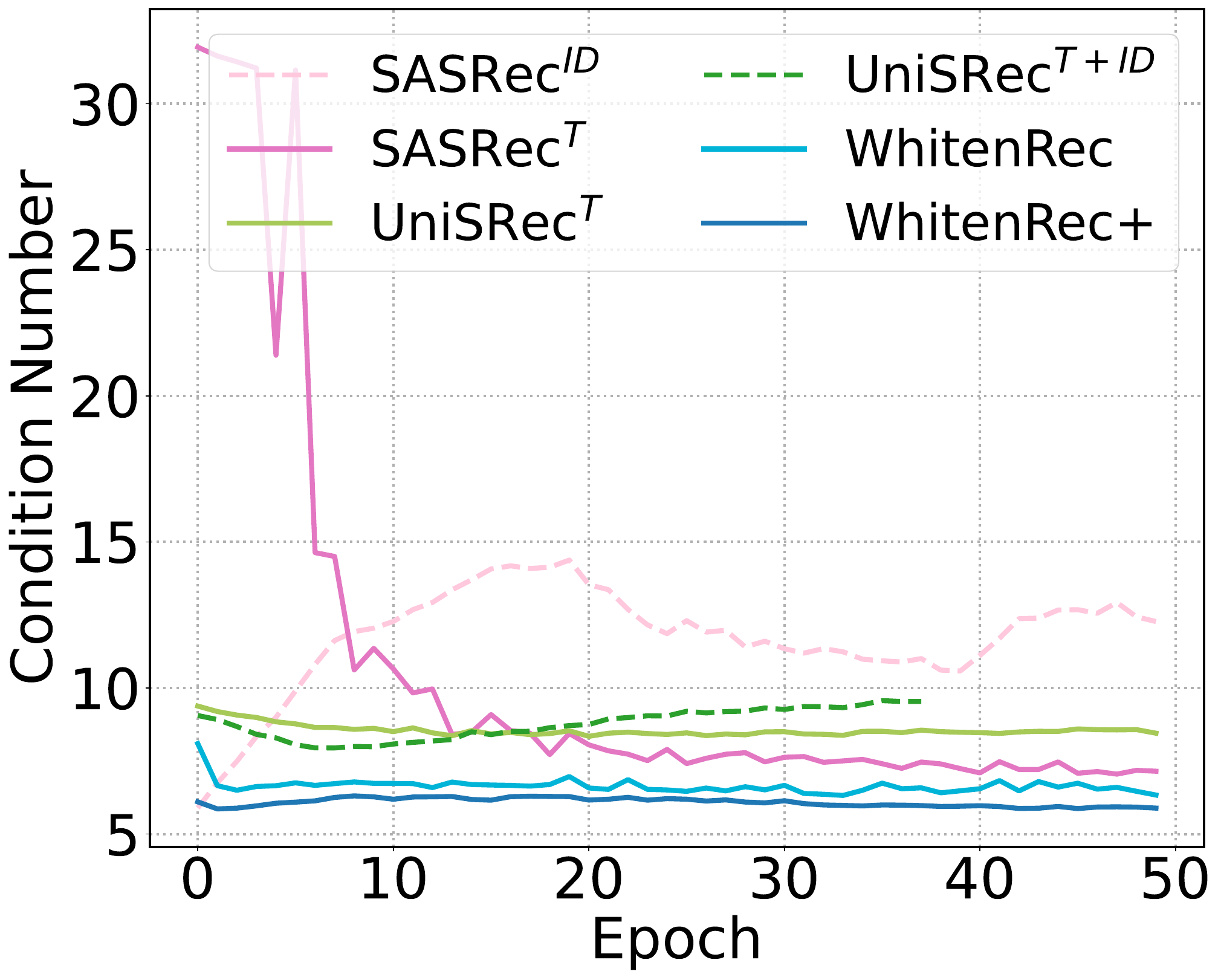}}
     \quad
     \subfloat[Tools-Condition Number]{\includegraphics[width=0.22\textwidth]{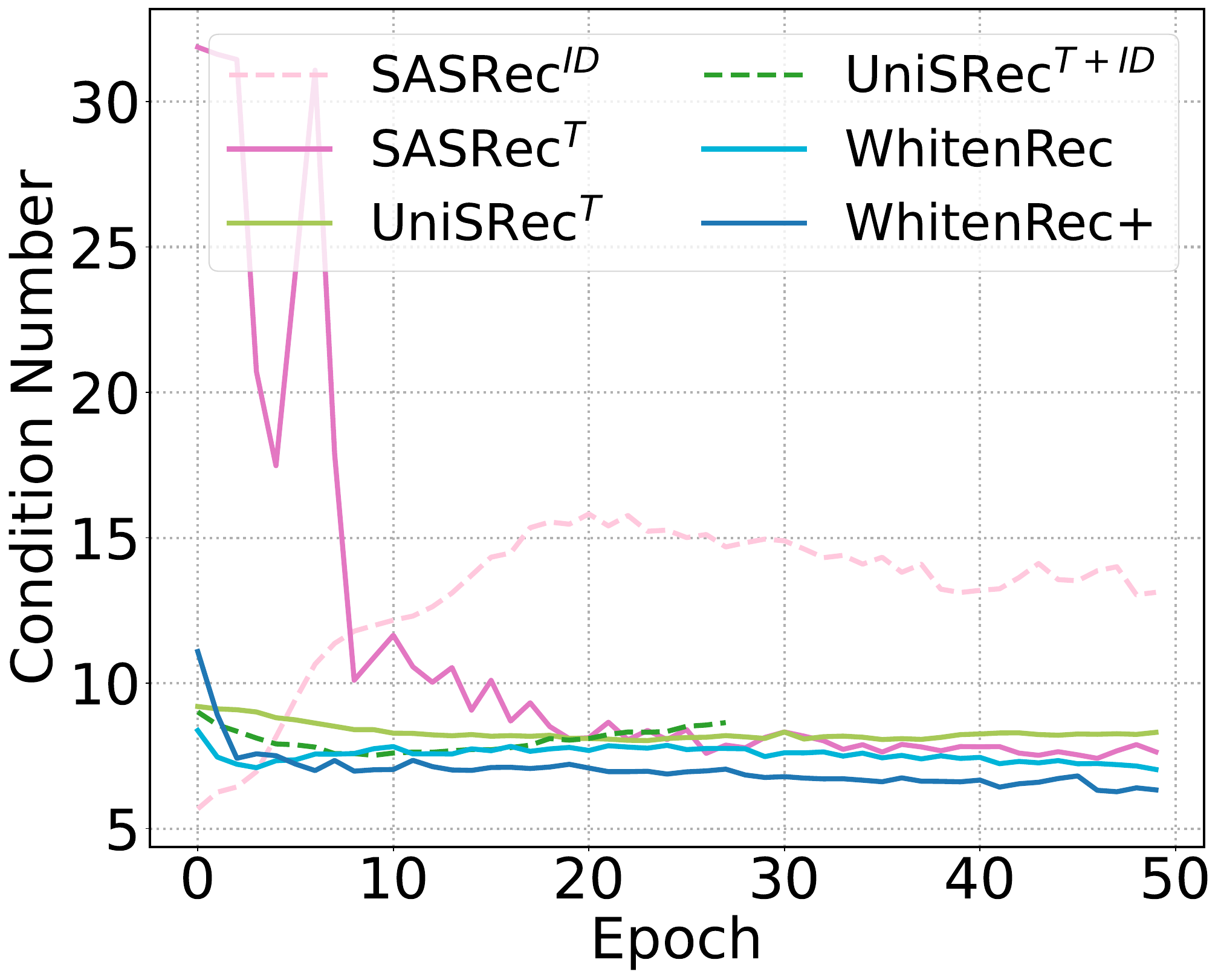}}
     \quad
     \subfloat[Food-Condition Number]{\includegraphics[width=0.22\textwidth]{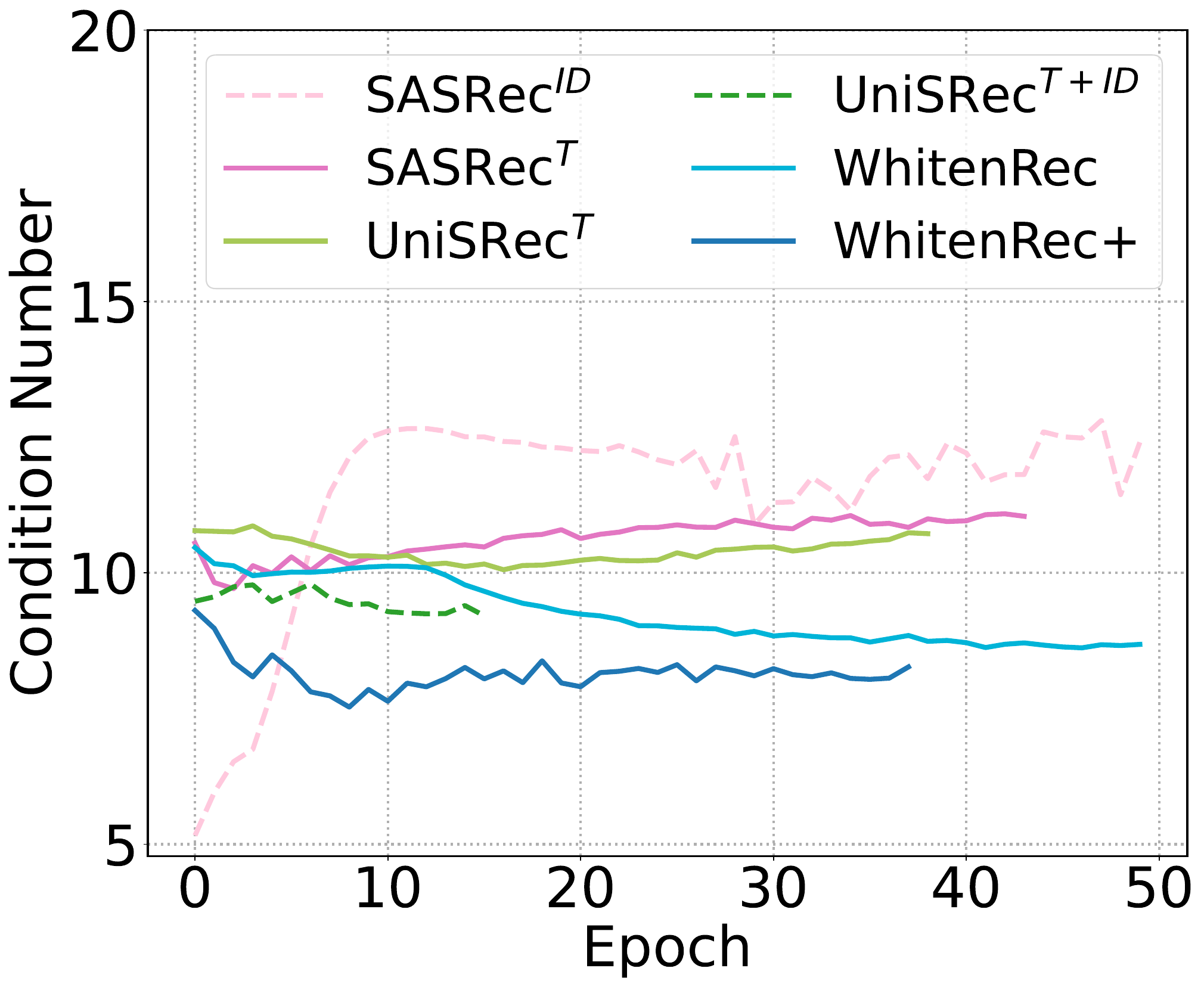}}
     \\
     \subfloat[Arts-Training Loss]{\includegraphics[width=0.22\textwidth]{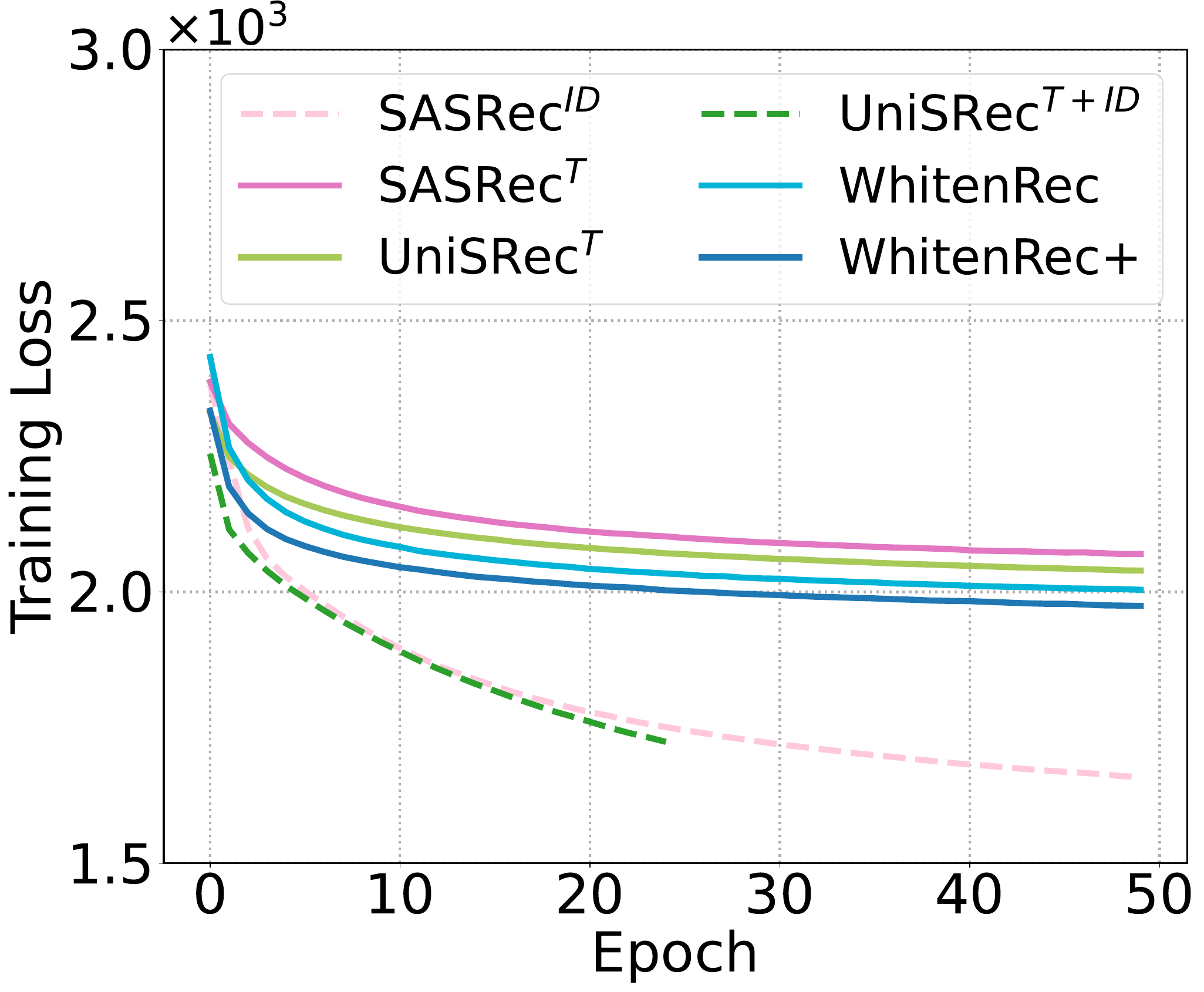}}
     \quad
     \subfloat[Toys-Training Loss]{\includegraphics[width=0.22\textwidth]{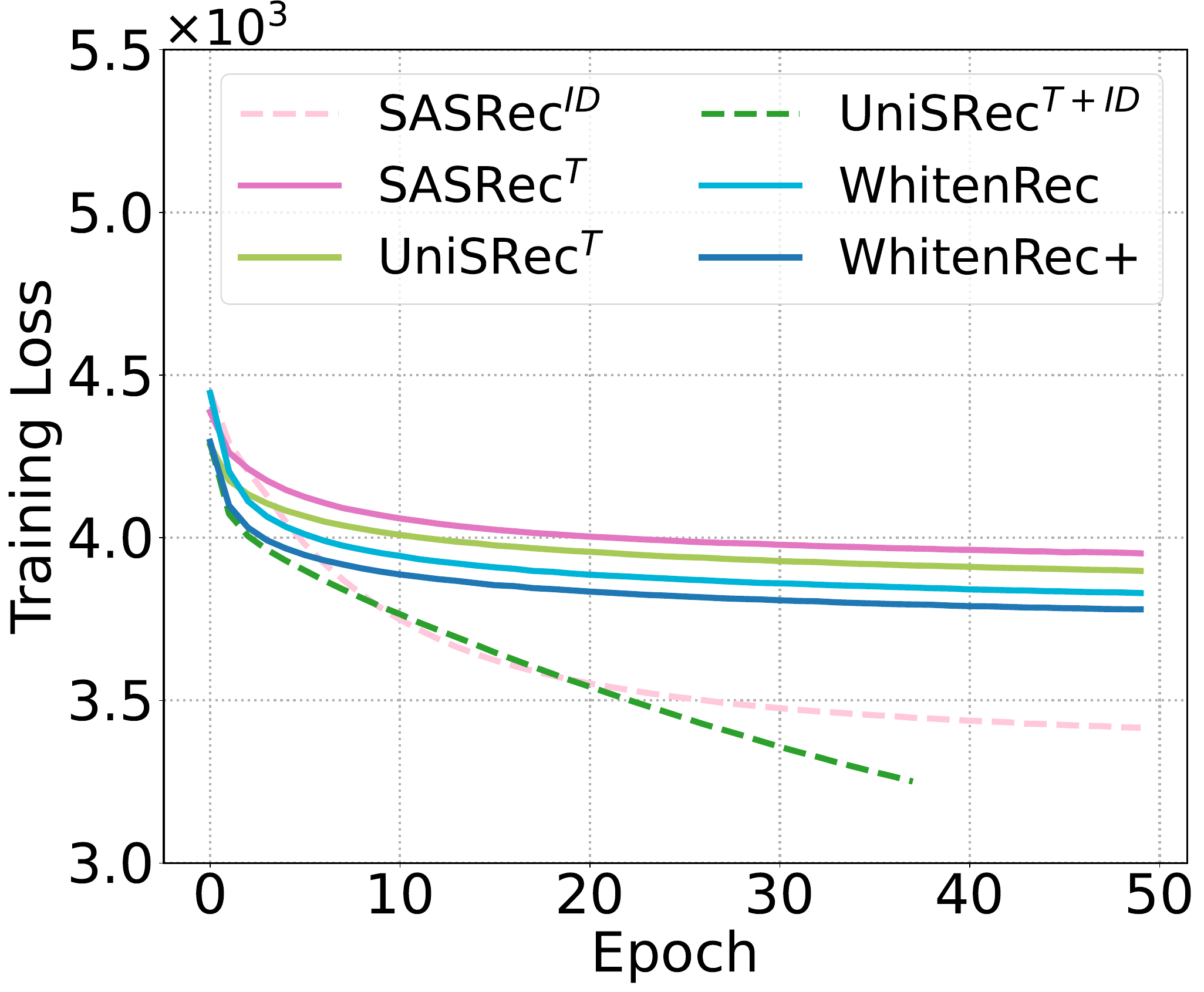}}
     \quad
     \subfloat[Tools-Training Loss]{\includegraphics[width=0.22\textwidth]{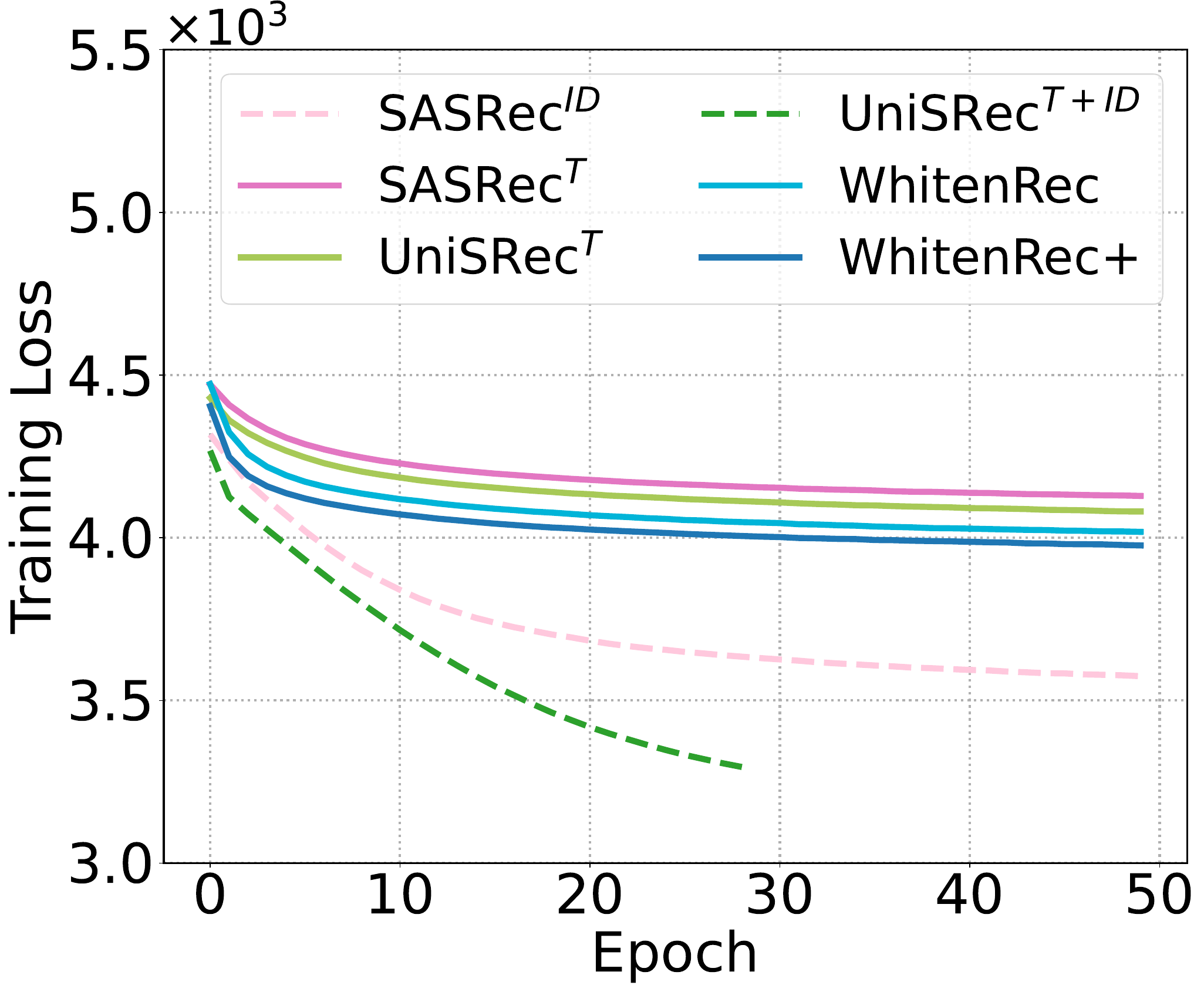}}
     \quad
     \subfloat[Food-Training Loss]{\includegraphics[width=0.22\textwidth]{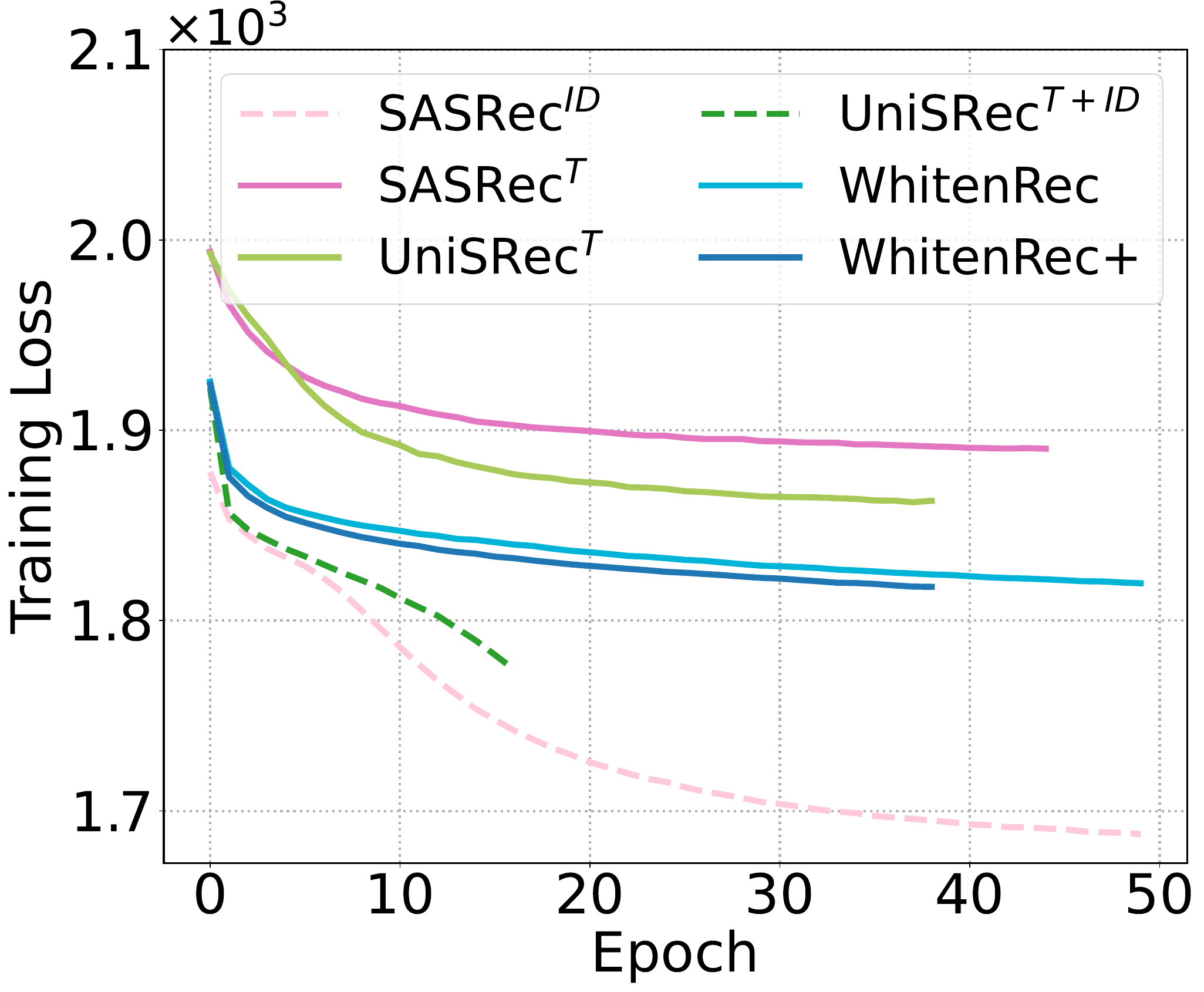}}
     \caption{\textcolor{black}{Conditioning analysis on Arts, Toys, Tools, and Food datasets. We plot the condition number (log-scale) calculated for the item embedding matrix after projection and the training loss with respect to each epoch. The dashed line represents a model that incorporates item ID embeddings.}}
     \label{fig:condition}
     \vspace{-1em}
\end{figure*}

We highlight the benefits of WhitenRec and WhitenRec+ in achieving improved conditioning of item embedding matrix $\mathbf{V}$. Given the covariance matrix $\mathbf{A}$ of $\mathbf{V}$, we measure its conditioning by the condition number~\cite{song2022improving}: $\kappa(\mathbf{A}) = \lambda_{max}(\mathbf{A}) \lambda_{min}^{-1}(\mathbf{A})$,
where $\lambda(\cdot)$ denotes the eigenvalue of the matrix.
Well-conditioned matrices have a low condition number, while ill-conditioned matrices have a high condition number. For neural networks, ill-conditioned covariance matrices cause detrimental effects on training stability and optimization. \textcolor{black}{Fig.~\ref{fig:condition}a-d show the evolution of the condition number throughout training epochs for Arts, Toys, Tools, and Food datasets. Fig.~\ref{fig:condition}e-h show the training loss over training epochs for the same datasets.}

\textcolor{black}{Among all text-based sequential recommendation models, both WhitenRec and WhitenRec+ converge more rapidly and achieve better conditioning compared to SASRec$^T$ or UniSRec$^T$. This outcome highlights the effectiveness of the whitening transformation in simplifying the optimization problem. Moreover, among these text-based methods, WhitenRec+ achieves the best conditioning and highest convergence rate.}
\textcolor{black}{In terms of training losses, methods that incorporate ID embeddings tend to decrease more rapidly, attributed to their larger parameter space, which can quickly capture important variance in the data. Yet, this increased parameter size can elevate the risk of overfitting, which could result in suboptimal performance. This tendency towards overfitting can be further evidenced by the worsening condition of these methods over time, as shown in Fig.~\ref{fig:condition}a-d, for most datasets.
Because overfitting can result in extremely large or small weights within the network, adversely affecting the condition number of the resultant item embedding matrix obtained from the network. }

\subsubsection{More Preserved Information in WhitenRec+}
We also conduct a mathematical proof to demonstrate that WhitenRec+ is capable of preserving more information than WhitenRec. 
Given a pre-trained text embedding matrix $\mathbf{X}\in\mathbb{R}^{d \times |\mathcal{I}|}$, where $n=|\mathcal{I}|$ is the number of items and $d$ is the dimension size.
We can derive the following proposition:
\begin{proposition}
WhitenRec+ preserves at least $(1-\frac{1}{G})d^2$ more information in its whitened representations compared to WhitenRec.
\end{proposition}
\begin{proof}
Given $\mathbf{X}$, we define the Gram matrix of $\mathbf{X}$ as $\mathbf{K}=\mathbf{X}^\top\mathbf{X} \in \mathbb{R}^{n \times n}$.
Based on~\cite{wadia2021whitening}, the prediction for recommendation models depends on training inputs only through $\mathbf{K}$\footnote{Pre-trained text embeddings are typically generated from a linear projector in conjunction with BERT.}.
Denote the text features whitened by WhitenRec as $\mathbf{Z}$. Since WhitenRec performs full data whitening on $\mathbf{X}$, we have $\mathbf{Z} \mathbf{Z}^\top = \mathbf{I}_{d}$. Thus, the Gram matrix of $\mathbf{Z}$ is:
\begin{align}
    \mathbf{K}_\mathbf{Z} &= \mathbf{Z}^\top \mathbf{Z} = \mathbf{Z}^\top \mathbf{Z} \mathbf{Z}^+ \mathbf{Z} = \mathbf{Z}^+ \mathbf{Z},
\end{align}
where $^+$ is the Moore–Penrose inverse and $\mathbf{Z} = \mathbf{Z} \mathbf{Z}^+ \mathbf{Z}$ holds true.

To determine the amount of decorrelated information preserved in $\mathbf{Z}$, we perform a transformation $\mathbf{Q}$ on $\mathbf{Z}$, resulting in $\widehat{\mathbf{Z}} = \mathbf{Q} \mathbf{Z} = \left[\mathbf{I} \cdots \right]$.
Here, $\mathbf{Q}\in\mathbb{R}^{d\times d}$ is the inverse of the submatrix formed by the first $d$ columns of $\widehat{\mathbf{Z}}$. As the first $d$ columns are deterministic, hence, $\widehat{\mathbf{Z}}$ can preserve $(n-d)d$ real values.
To reconstruct ${\mathbf{K}}_\mathbf{Z}$, 
\begin{align}
    {\mathbf{K}}_\mathbf{Z} &= {\mathbf{Z}}^+ {\mathbf{Z}} = {\mathbf{Z}}^+ \mathbf{Q}^{-1} \mathbf{Q} {\mathbf{Z}} = \pp{\mathbf{Q} {\mathbf{Z}}}^+ \mathbf{Q} {\mathbf{Z}}
    = \widehat{\mathbf{Z}}^+ \widehat{\mathbf{Z}}.
\end{align}
The above equation indicates that both the whitened text feature matrix $\mathbf{Z}$ and $\widehat{\mathbf{Z}}$ contain an equivalent amount of information for reconstructing $\mathbf{K}_\mathbf{Z}$.

Analogously, we can infer a submatrix (one group) whitened by WhitenRec+ preserves $(n-\frac{d}{G})\frac{d}{G}$ real values. With $G$ groups of submatrices, thus, WhitenRec+ has $(n-\frac{d}{G})d$ values to reconstruct ${\mathbf{K}}$. 
Consequently, WhitenRec+, which integrates both a single group and $G$ groups of whitened representations, retains a minimum of $(1-\frac{1}{G})d^2=(n-\frac{d}{G})d - (n-d)d$ additional information in enhancing the training of the model.
\end{proof}

\subsection{\textcolor{black}{Complexity Analysis}}

\textcolor{black}{Note that various degrees of the whitening transformation can be pre-computed; hence, WhitenRec and WhitenRec+ exhibit identical complexity levels. The time complexity of our models primarily arises from the projection head with MLPs and the attention-based transformer layers. Each contributes a time complexity of \(\mathcal{O}(|s|(d_td+d^2))\) and \(\mathcal{O}(|s|^2d+|s|d^2)\), respectively. Consequently, the aggregate time complexity is \(\mathcal{O}(|s|d_td+|s|d^2+|s|^2d)\). 
In the experimental section, we show that our methods have significantly reduced the number of parameters and training time in practice.}


\section{Experiments}

\subsection{Experimental Settings}

\subsubsection{Datasets}
To evaluate the performance of the proposed method, we conduct experiments on four real-world datasets that are commonly used in evaluating recommender systems. Three of them are representative categories from the Amazon review dataset~\cite{ni2019justifying}: \textit{Arts, Crafts and Sewing}, \textit{Toys and Games}, and \textit{Tools and Instruments}. We abbreviate them as Arts, Toys, and Tools. \textcolor{black}{Another one is Food\footnote{https://www.kaggle.com/datasets/elisaxxygao/foodrecsysv1}, which collects food recipes from Allrecipes.com.}

\begin{table}
\caption{Dataset Statistics. “Avg. n” and “Avg. i” denote the average length of user sequences and the average actions of items.}
  \centering
   \small
  \begin{tabular}{l|cccccc}
    \toprule
    Datasets & \#Users & \#Items & \#Inter. & Avg. n & Avg. i\\ 
    \midrule
    Arts   & 45,486 & 21,019 & 349,664 & 7.69 & 16.63 \\
    Toys &  85,694 & 40,483 & 618,738 & 7.22 & 15.28 \\
    Tools &  90,599 & 36,244 & 623,248 & 6.88 &17.20 \\
    \textcolor{black}{Food} & \textcolor{black}{28,988} & \textcolor{black}{12,910} & \textcolor{black}{274,509} & \textcolor{black}{9.47} & \textcolor{black}{21.26} \\
    \bottomrule
  \end{tabular}
  \label{tab:stats}
  \vspace{-1.2em}
\end{table}

\subsubsection{Baseline Methods}

To evaluate the effectiveness of our proposed method, we compare it with several state-of-the-art recommendation models. These baselines fall into three groups: general recommendation models with text features (\ie GRCN, BM3), sequential recommendation models (\ie SASRec$^{ID}$, CL4SRec), and sequential recommendation models with text features (\ie SASRec$^{T}$, SASRec$^{T+ID}$, S$^3$-Rec, FDSA, UniSRec$^{T}$, UniSRec$^{T+ID}$).
This evaluation excludes general recommendation models based solely on user-item interactions~\cite{zhang2022diffusion, zhou2023selfcf, zhou2023layer}, as these are generally outperformed by the aforementioned models.
\textbf{GRCN}~\cite{wei2020graph} is a graph-based multimodal recommendation model that refines the user-item interaction graph by identifying false-positive feedback and pruning noisy edges. Only item text representations are exploited;
\textbf{BM3}~\cite{zhou2022bootstrap} utilizes contrastive learning losses for multimodal recommendation. Only item text representations are exploited;
\textbf{SASRec$^{ID}$}~\cite{kang2018self} is a directional self-attention method for next item prediction. Item representations are randomly initialized ID embeddings; 
\textbf{CL4SRec}~\cite{xie2022contrastive} designs three data augmentation approaches to construct contrastive tasks and extract self-supervised signals to improve the sequential recommendation performance;
\textbf{SASRec$^{T}$} and \textbf{SASRec$^{T+ID}$} are two extensions of SASRec~\cite{kang2018self}. \textbf{SASRec}$^{T}$ transforms item text representations with MLPs as the input for the self-attention blocks. \textbf{SASRec$^{T+ID}$} combines item ID embeddings with text representations of items, transformed by MLPs, as the input for the self-attention blocks;
\textbf{S$^3$-Rec}~\cite{zhou2020s3} devises supervised learning objectives to learn the correlations between items and features;
\textbf{FDSA}~\cite{zhang2019feature} models the transition patterns between items as well as features by separate self-attention blocks;
\textbf{UniSRec}~\cite{hou2022towards} leverages item text representations with an MoE-based adaptor and employs contrastive learning tasks to learn transferable sequence representations. For a fair comparison, we remove its pre-training stage and fine-tune the model with the inductive setting and the transductive setting, which are denoted as \textbf{UniSRec$^{T}$} and \textbf{UniSRec$^{T+ID}$} respectively. The inductive setting takes into account only item text representations, whereas the transductive setting takes into account both item text and ID representations. 
\textcolor{black}{\textbf{VQRec}~\cite{hou2023learning} proposes to transform text encodings into discrete codes, followed by utilizing the embedding lookup for refining item textual representation from pre-trained language models. For a fair comparison, we remove its pre-training stage and directly fine-tune the model with vector-quantized item representations.}

\subsubsection{Evaluations}
We conduct experiments in both warm-start and cold-start settings.

\textbf{Warm-start settings.} Following~\cite{zhou2020s3,sun2019bert4rec}, we keep the five-core datasets and discard users and items with fewer than five interactions. We apply the \textit{leave-one-out} strategy to evaluate the performance. Specifically, for each user, the last item of her interaction sequence is used for testing, the second last item is used for validation, and the remaining items are used for model training. \textcolor{black}{For all experiments, unless otherwise stated, they are conducted under the warm-start setting.}

\textbf{Cold-start settings.}
Following \cite{wei2021contrastive}, a subset of items (15\% of all items) is randomly selected, and all user-item interactions related to this subset are removed. We preserve sequences containing the aforementioned “cold” items as target items in the validation and testing sets. 
Since these items are not encountered by the model during training, we can assess the model's capability to generalize to previously unseen items.

Each method is evaluated on the entire item set without sampling to avoid inconsistent results~\cite{krichene2020sampled}. The recommendation performance is evaluated by two widely used metrics, \ie Recall@$K$ and Normalized Discounted Cumulative Gain@$K$ (respectively denoted by R@$K$ and N@$K$). In the experiments, $K$ is empirically set to 20 and 50.

\subsubsection{Implementation Details}
The proposed method is implemented by Pytorch~\cite{paszke2019pytorch} and an open-source recommendation framework RecBole~\cite{zhao2021recbole}. The Adam optimizer~\cite{kingma2014adam} is used to learn model parameters. 
For a fair comparison with baselines, we set the maximum sequence length, embedding size, and batch size to $50$, $300$, and $1024$, respectively. We fix the number of self-attention blocks, attention heads, and MLP layers used in the projection head at 2. 
Other hyper-parameters of baseline methods are selected following the original papers, and the optimal settings are chosen based on the model performance on validation data. For our proposed methods, we tune the learning rate in $\{1e^{-5},1e^{-4}, 5e^{-4}, 1e^{-3}\}$ and weight decay in $\{0, 1e^{-4}, 1e^{-6}\}$. We adopt an early stopping strategy, \ie we apply a premature stopping if N@20 on the validation data does not increase for 10 epochs to avoid over-fitting.

\begin{table*}
    \caption{Performance comparison of different methods on the warm-start setting. The best results are in \textbf{boldface}, and the best results for baselines are \underline{underlined}. * denotes WhitenRec or WhitenRec+ surpasses the best baseline using a paired t-test ($p < 0.01$). The features utilized for item representations in each model are categorized as ID, text (T), or a combination of both (T+ID).}
    \centering
    \small
    \resizebox{\textwidth}{!}{
    \begin{tabular}{l@{\hspace{0.5\tabcolsep}}c|@{\hspace{0.5\tabcolsep}} c@{\hspace{0.5\tabcolsep}} c|@{\hspace{0.5\tabcolsep}} c@{\hspace{0.5\tabcolsep}} c@{\hspace{0.5\tabcolsep}}| c@{\hspace{0.5\tabcolsep}} c@{\hspace{0.5\tabcolsep}} c@{\hspace{0.5\tabcolsep}} c@{\hspace{0.5\tabcolsep}} c@{\hspace{0.5\tabcolsep}} c@{\hspace{0.5\tabcolsep}} c @{\hspace{0.5\tabcolsep}} |c@{\hspace{0.5\tabcolsep}} c@{\hspace{0.5\tabcolsep}} }
    \toprule
    \multirow{2}{*}{Dataset} & \multirow{2}{*}{Metric}
    &GRCN&BM3&SASRec&CL4SRec&SASRec&SASRec&S$^3$-Rec&FDSA&UniSRec&UniSRec&\textcolor{black}{VQRec}&WhitenRec&WhitenRec+\\ 
    &&(T+ID)&(T+ID)&(ID)&(ID)&(T)&(T+ID)&(T+ID)&(T+ID)&(T)&(T+ID)&\textcolor{black}{(T)}&(T)&(T)\\
    
    \midrule
    \multirow{4}{*}{Arts}
    &R@20 &0.0851&0.1233&0.1410&0.1388&0.1476&0.1435&0.1411&0.1284&0.1500&\underline{0.1611}&\textcolor{black}{0.1390}&0.1625&\textbf{0.1688}*\\
    & R@50 &0.1296&0.1782&0.1967&0.1967&0.2129&0.2009&0.2007&0.1788&0.2165&\underline{0.2322}&\textcolor{black}{0.1947}&0.2348&\textbf{0.2403}*\\
    & N@20 &0.0411&0.0642&0.0776&0.0653&0.0721&0.0766&0.0762&\underline{0.0785}&0.0738&0.0774&\textcolor{black}{0.0734}&0.0796&\textbf{0.0810}*\\
    & N@50 &0.0499&0.0750&0.0887&0.0768&0.0850&0.0879&0.088&0.0888&0.0869&\underline{0.0915}&\textcolor{black}{0.0843}&0.0939*&\textbf{0.0952}*\\

    \midrule
    \multirow{4}{*}{Toys}
    & R@20 &0.0651&0.0965&0.1121&0.1094&0.0983&\underline{0.1163}&0.1068&0.0895&0.1042&\textbf{0.1257}&\textcolor{black}{0.1075}&0.1201&\textbf{0.1257}\\
    & R@50 &0.0981&0.1383&0.1581&0.1609&0.1542&0.1664&0.1533&0.1242&0.1607&\underline{0.1801}&\textcolor{black}{0.1491}&0.1798&\textbf{0.1874}*\\
    & N@20 &0.0304&0.0478&0.0467&0.0426&0.0429&0.0511&0.0488&0.0475&0.0451&\underline{0.0513}&\textcolor{black}{0.0468}&0.0521&\textbf{0.0537}*\\
    & N@50 &0.0369&0.0560&0.0558&0.0528&0.0539&0.0610&0.0581&0.0543&0.0563&\underline{0.0621}&\textcolor{black}{0.0550}&0.0639*&\textbf{0.0659}*\\
    
    \midrule
    \multirow{4}{*}{Tools}
    & R@20 &0.0452&0.0530&0.0712&0.0781&0.0739&0.0728&0.0707&0.0633&0.0772&\underline{0.0828}&\textcolor{black}{0.0734}&0.0861*&\textbf{0.0888}*\\
    & R@50 &0.0682&0.0714&0.0941&0.1027&0.1055&0.0954&0.0943&0.0812&0.1091&\underline{0.1116}&\textcolor{black}{0.0963}&0.1196*&\textbf{0.1236}*\\
    & N@20 &0.0234&0.0299&0.0418&0.0385&0.0386&\underline{0.0445}&0.0424&0.0432&0.0407&0.0420&\textcolor{black}{0.0423}&0.0453&\textbf{0.0462}*\\
    & N@50&0.0280&0.0335&0.0463&0.0433&0.0448&\underline{0.0490}&0.0470&0.0468&0.0470&0.0477&\textcolor{black}{0.0468}&0.0519*&\textbf{0.0531}*\\

    \midrule
    \multirow{4}{*}{\textcolor{black}{Food}}
& \textcolor{black}{R@20} &\textcolor{black}{0.0408}&\textcolor{black}{0.0540}&\textcolor{black}{0.0520}&\textcolor{black}{0.0531}&\textcolor{black}{0.0541}&\textcolor{black}{0.0547}&\textcolor{black}{0.0522}&\textcolor{black}{0.0518}&\textcolor{black}{0.0544}&\textcolor{black}{\underline{0.0555}}&\textcolor{black}{0.0471}&\textcolor{black}{0.0569}&\textcolor{black}{\textbf{0.0586}*}\\
& \textcolor{black}{R@50} &\textcolor{black}{0.0796}&\textcolor{black}{0.0947}&\textcolor{black}{0.0955}&\textcolor{black}{0.0949}&\textcolor{black}{0.0991}&\textcolor{black}{0.0984}&\textcolor{black}{0.0960}&\textcolor{black}{0.0960}&\textcolor{black}{\underline{0.1018}}&\textcolor{black}{0.1001}&\textcolor{black}{0.0861}&\textcolor{black}{0.1043}&\textcolor{black}{\textbf{0.1072}*}\\
& \textcolor{black}{N@20} &\textcolor{black}{0.0162}&\textcolor{black}{0.0215}&\textcolor{black}{0.0208}&\textcolor{black}{0.0214}&\textcolor{black}{0.0220}&\textcolor{black}{0.0222}&\textcolor{black}{0.0210}&\textcolor{black}{0.0210}&\textcolor{black}{0.0221}&\textcolor{black}{\underline{0.0223}}&\textcolor{black}{0.0189}&\textcolor{black}{0.0227}&\textcolor{black}{\textbf{0.0234}*}\\
& \textcolor{black}{N@50} &\textcolor{black}{0.0239}&\textcolor{black}{0.0295}&\textcolor{black}{0.0294}&\textcolor{black}{0.0297}&\textcolor{black}{0.0308}&\textcolor{black}{0.0308}&\textcolor{black}{0.0296}&\textcolor{black}{0.0297}&\textcolor{black}{\underline{0.0315}}&\textcolor{black}{0.0311}&\textcolor{black}{0.0266}&\textcolor{black}{0.0321}&\textcolor{black}{\textbf{0.0330}*}\\

    \bottomrule
    \end{tabular}}
    \vspace{-1em}
    \label{tab:overall}
\end{table*}

\begin{table*}
    \caption{Performance comparison of different methods on the cold-start setting. The best results are in \textbf{boldface}, and the second best results are \underline{underlined}. * denotes WhitenRec or WhitenRec+ surpasses the best baseline using a paired t-test ($p < 0.01$).}
    \centering
    \small
    \begin{tabular}{l|cc|cc|c c| cc}
    \toprule
    \multirow{2}{*}{Model}&\multicolumn{2}{c|}{Arts}&\multicolumn{2}{c|}{Toys}
&\multicolumn{2}{c|}{Tools} &\multicolumn{2}{c}{\textcolor{black}{Food}}\\   
    &R@20&N@20&R@20&N@20&R@20&N@20&\textcolor{black}{R@20}&\textcolor{black}{N@20}\\
    \midrule
    SASRec(T)&0.0300&0.0130&0.0239&0.0100&0.0153&0.0057&\textcolor{black}{0.0031}&\textcolor{black}{0.0013}\\
    UniSRec(T)&0.0617&0.0281&0.0519&0.0222&0.0298&0.0158&\textcolor{black}{0.0037}&\textcolor{black}{0.0011}\\
    \midrule
    WhitenRec$_{G=1}$(T)&0.0554&0.0271&0.0530&0.0238*&0.0431*&0.0234*&\textcolor{black}{0.0037}&\textcolor{black}{0.0012}\\
    WhitenRec$_{G>1}$(T)&\underline{0.0656}*&\underline{0.0297}*&\underline{0.0624}*&\underline{0.0265}*&\underline{0.0501}*&\underline{0.0252}*&\underline{\textcolor{black}{0.0044*}}&\underline{\textcolor{black}{0.0014*}}\\
    WhitenRec+(T)&\textbf{0.0693}*&\textbf{0.0315}*&\textbf{0.0626}*&\textbf{0.0266}*&\textbf{0.0537}*&\textbf{0.0268}*&\textbf{\textcolor{black}{0.0048*}}&\textbf{\textcolor{black}{0.0017*}}\\
    \bottomrule
    \end{tabular}
    \label{tab:cold}
    \vspace{-1em}
\end{table*}
\subsection{Performance Comparison}
\label{sec:exp}

\subsubsection{Overall performance}
Table~\ref{tab:overall} shows the performance comparison results for warm-start settings, from which we can observe:
(1) General recommendation methods utilizing text features perform worse than sequential methods for three Amazon datasets, highlighting the effectiveness of sequence encoders in capturing sequential data patterns. \textcolor{black}{However, for the Food dataset, a more advanced general recommendation strategy that utilizes text features, specifically BM3, either excels or matches the performance of sequential methods. This observation indicates that text features in Food dataset contain more pertinent information, which significantly contributes to the enhancement of recommendation accuracy.}
(2) Sequential methods utilizing text features yield better performance overall, suggesting that text features provide rich semantic information about items, and can enhance recommendation accuracy.
(3) Our methods WhitenRec and WhitenRec+ significantly outperform both general recommendation methods with text features and sequential recommendation methods, demonstrating the effectiveness of the whitening for text features extracted from pre-trained encoders. 
(4) The performance of both WhitenRec and WhitenRec+ is comparable or superior to that of SASRec$^{T+ID}$ or UniSRec$^{T+ID}$.This finding suggests that the proposed whitening transformation approach can achieve improved results without depending on ID embeddings, while also reducing the number of learnable parameters.
(5) WhitenRec+ can further improve recommendation performance compared with WhitenRec. This indicates leveraging both fully whitened and relaxed whitened text representations can enhance the item representation learning, and therefore improve the sequential recommendation performance.

\subsubsection{Performance in cold-start settings}
The cold-start problem persists in recommender systems. Item text features provide rich content information that can alleviate the cold-start problem. 
\textcolor{black}{Since our objective is to predict user preferences for items not present in the training data, item ID embeddings become unlearnable and thus cannot be utilized during inference. 
Given this context, SASRec\(^T\) and UniSRec\(^T\) have been selected as representative baseline models. Because they solely depend on item text features, which aligns with the constraints of the cold-start setting. Furthermore, these models exhibit relatively better performance compared to other baselines.} 
Table~\ref{tab:cold} shows the results of the performance comparison, from which we can observe:
(1) UniSRec$^{T}$ performs better than SASRec$^{T}$, indicating the effectiveness of utilizing the Mixture-of-Experts adaptor with parametric whitening to transform text embeddings for the recommendation task.
(2) Full whitening WhitenRec$_{G=1}$ is either surpassed by or yields similar performance to UniSRec$^{T}$ in the Arts and Toys datasets. In contrast, relaxed whitening WhitenRec$_{G>1}$ outperforms WhitenRec$_{G=1}$ and baselines. It suggests that the utilization of relaxed whitened representations facilitates improved generalization for unseen data, ultimately leading to greater performance enhancement.
(3) Our proposed method, WhitenRec+, demonstrates the best performance across all baselines for all three datasets. 
Leveraging both full and relaxed whitening transformation on text features is proved to be effective under the cold-start setting.

\subsection{Effect of Group Size}
\begin{figure*}
     \centering
     \subfloat[Arts-R@20]{\includegraphics[width=0.22\textwidth]{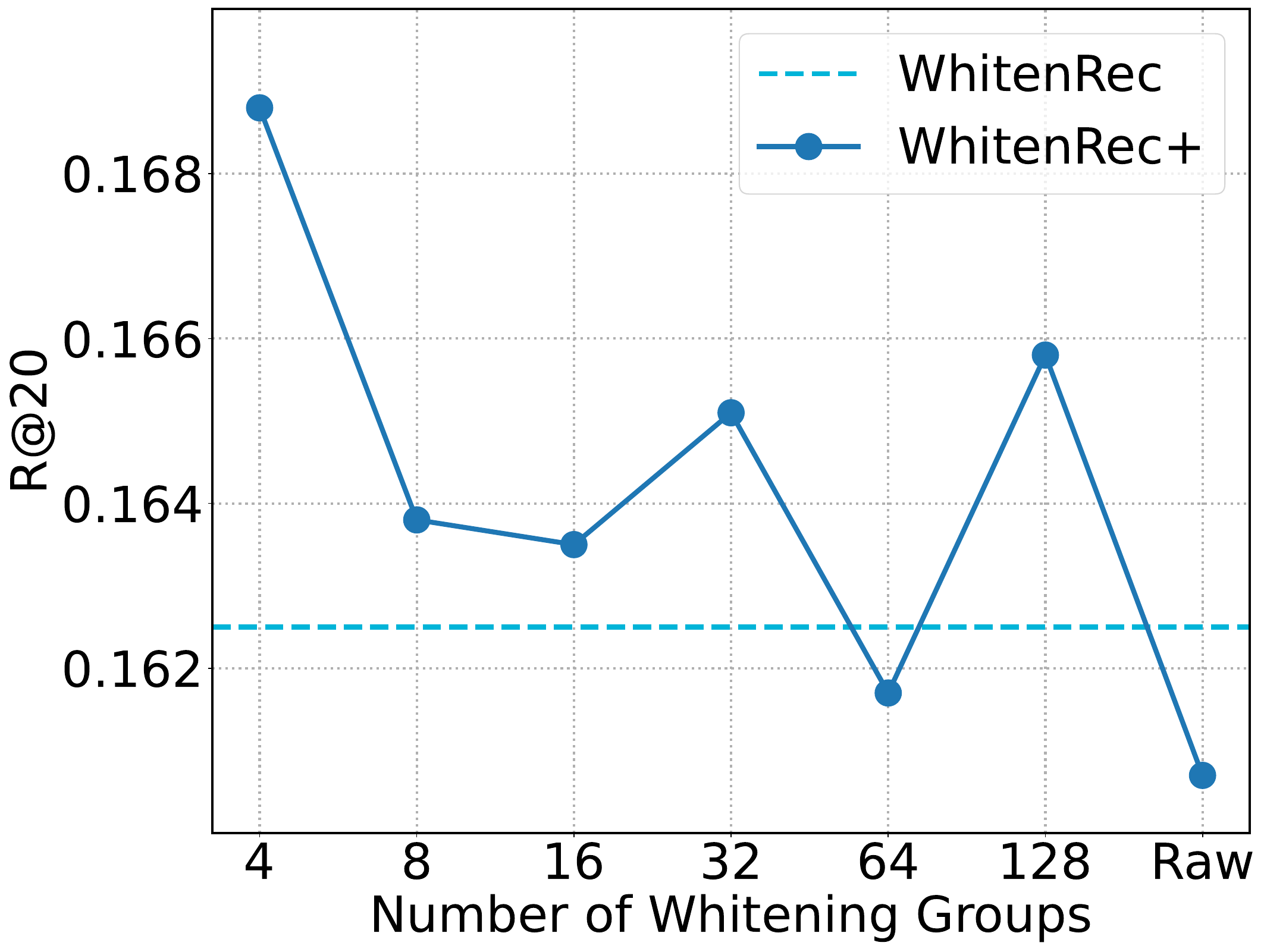}}
     \quad
    \subfloat[Toys-R@20]{\includegraphics[width=0.22\textwidth]{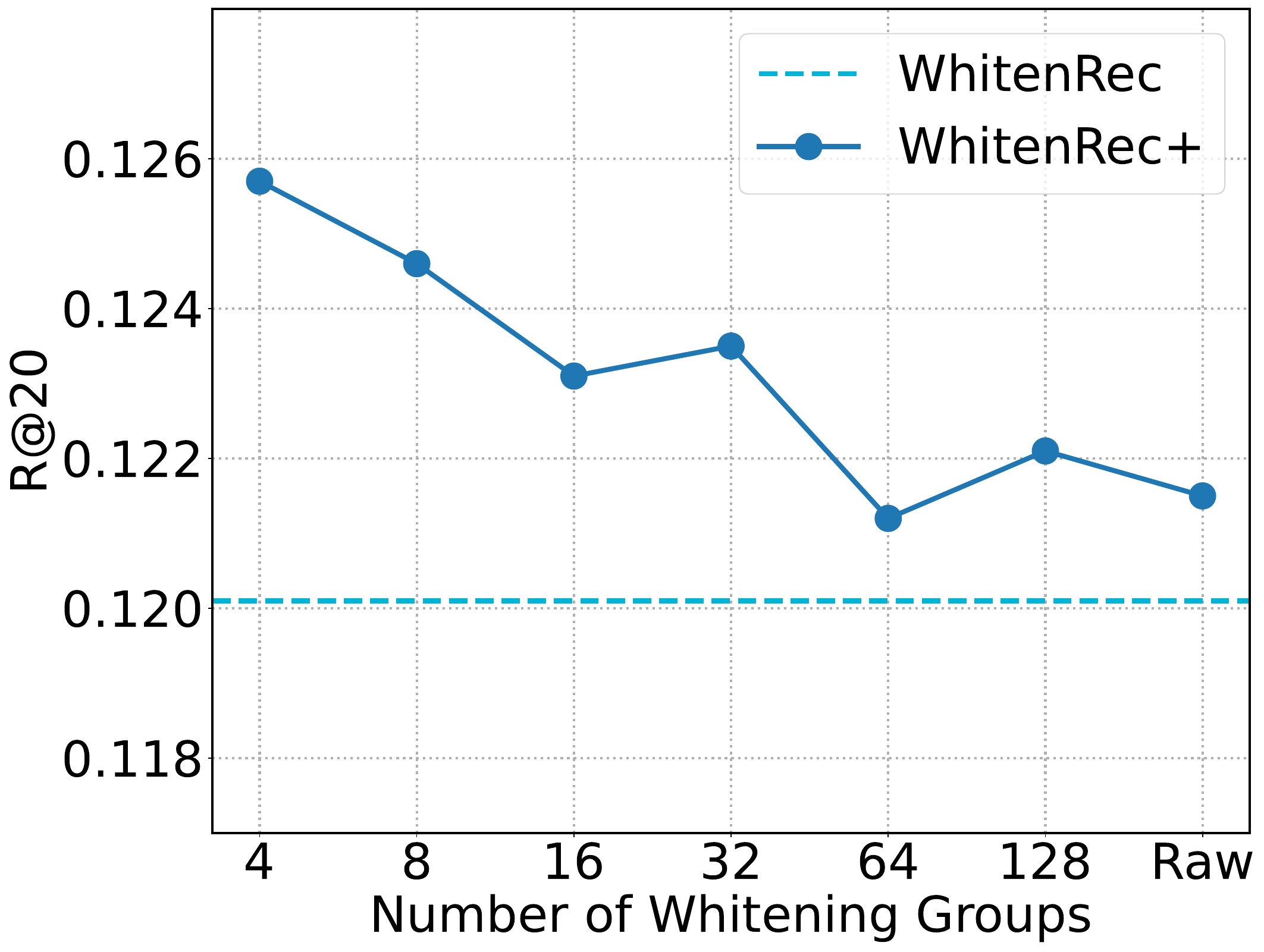}}
    \quad
    \subfloat[Tools-R@20]{\includegraphics[width=0.22\textwidth]{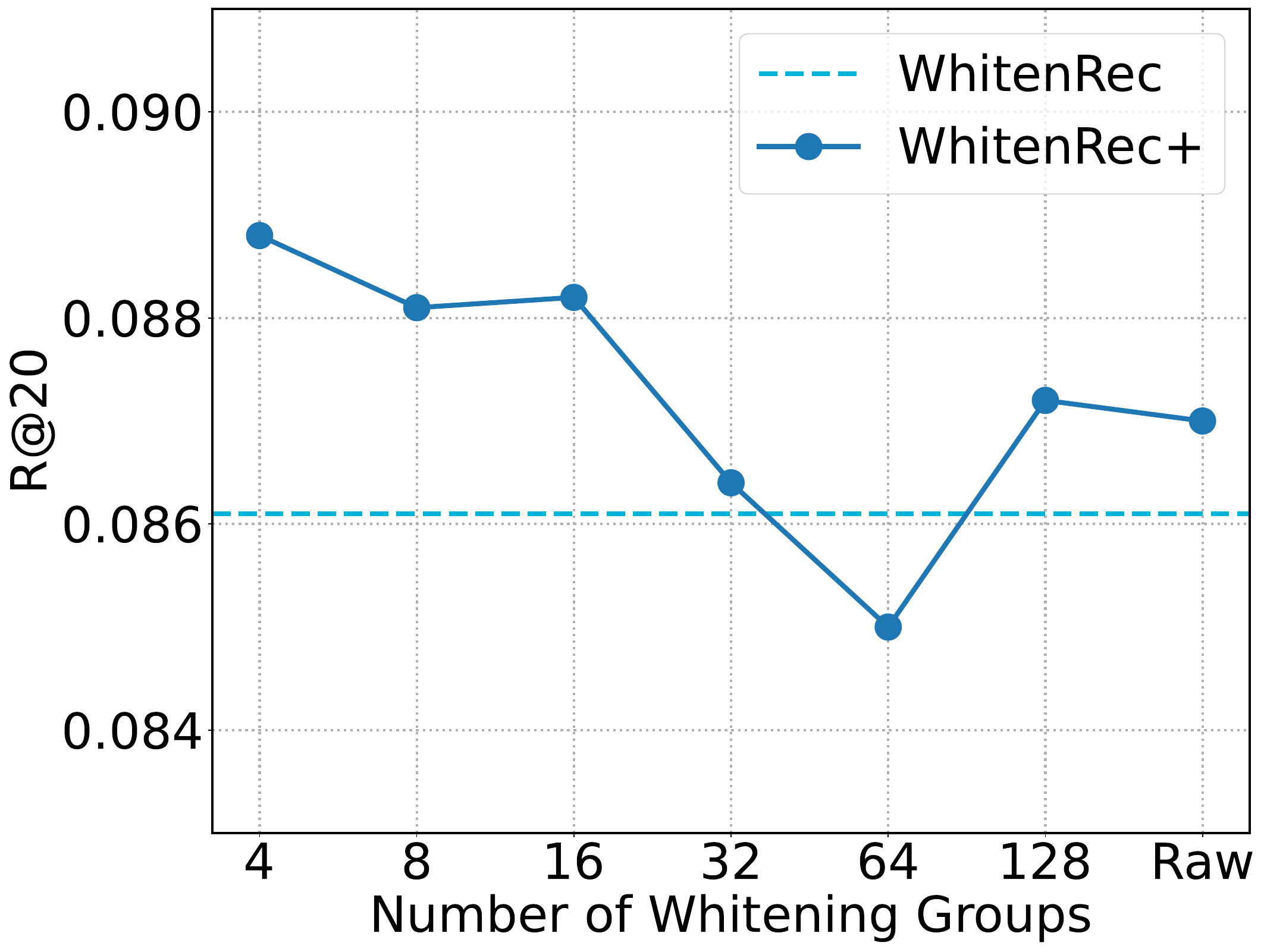}}
    \quad
    \subfloat[\textcolor{black}{Food-R@20}]{\includegraphics[width=0.22\textwidth]{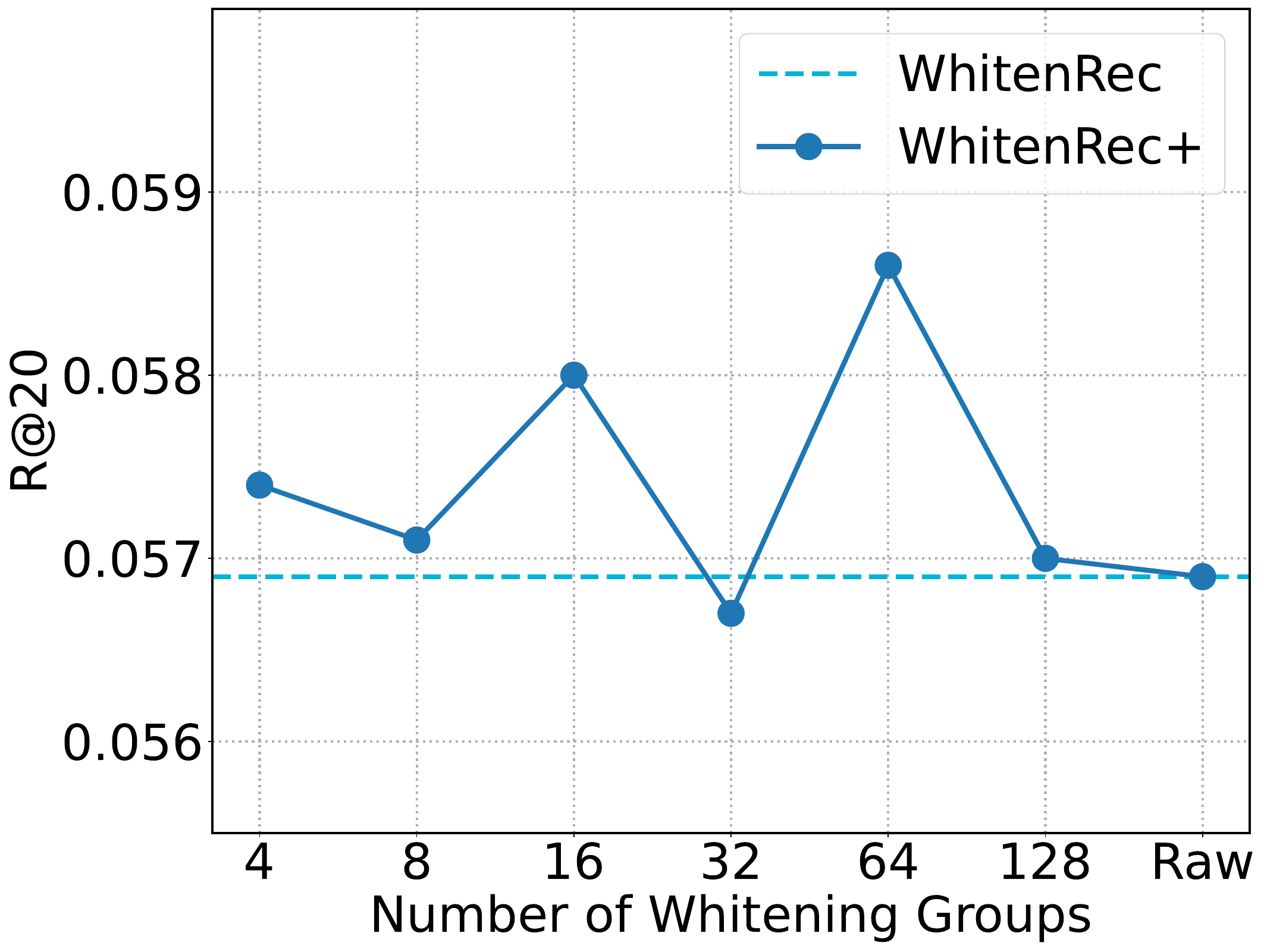}}
     \caption{Performance by Whitening Groups for WhitenRec+. }
     \label{fig:group}
     \vspace{-1em}
\end{figure*}

\begin{table*}
    \caption{Performance comparison of projection head for WhitenRec+.}
    \centering
    \small
    \begin{tabular}{l| c c |c c| c c| c c}
     \toprule
    \multirow{2}{*}{Model}&\multicolumn{2}{c|}{Arts}&\multicolumn{2}{c|}{Toys}&\multicolumn{2}{c|}{Tools} &\multicolumn{2}{c}{\textcolor{black}{Food}}\\   
    &R@20&N@20&R@20&N@20&R@20&N@20&\textcolor{black}{R@20}&\textcolor{black}{N@20}\\
    \midrule
     Linear&0.1476&0.0724&0.1029&0.0448&0.0751&0.0396&\textcolor{black}{0.0551}&\textcolor{black}{0.0222}\\
    MLP-1&0.1627&0.0782&0.1168&0.0494&0.0836&0.0427&\textcolor{black}{0.0560}&\textcolor{black}{0.0221}\\
    MLP-2&\underline{0.1688}&\textbf{0.0810}&\underline{0.1257}&\underline{0.0537}&\underline{0.0888}&\underline{0.0462}&\textcolor{black}{\textbf{0.0586}}&\textcolor{black}{\textbf{0.0234}}\\
    MLP-3&{0.1655}&\underline{0.0808}&\textbf{0.1261}&\textbf{0.0547}&\textbf{0.0894}&\textbf{0.0469}&\textcolor{black}{\underline{0.0565}}&\textcolor{black}{\underline{0.0229}}\\
    \textcolor{black}{MoE} & \textcolor{black}{\textbf{0.1690}} & \textcolor{black}{0.0784} & \textcolor{black}{0.0896} & \textcolor{black}{0.0446} & \textcolor{black}{0.0852} & \textcolor{black}{0.0438} & \textcolor{black}{0.0553} & \textcolor{black}{0.0221}\\
    \bottomrule
    \end{tabular}
    \vspace{-1em}
    \label{tab:head}
\end{table*}

\begin{table*}
    \caption{Performance comparison of whitening methods for WhitenRec+.}
    \centering
    \small
    \begin{tabular}{l| c c| c c| c c|cc}
     \toprule
    \multirow{2}{*}{Model}&\multicolumn{2}{c|}{Arts}&\multicolumn{2}{c|}{Toys}&\multicolumn{2}{c|}{Tools} &\multicolumn{2}{c}{\textcolor{black}{Food}}\\   
    &R@20&N@20&R@20&N@20&R@20&N@20&\textcolor{black}{R@20}&\textcolor{black}{N@20}\\
    \midrule
    PW &0.1243&0.0599 &0.0843&0.0363&0.0626&0.0322&\textcolor{black}{0.0547}&\textcolor{black}{0.0217}\\
    \textcolor{black}{BERT-flow}&\textcolor{black}{0.1550}&\textcolor{black}{0.0755}&\textcolor{black}{0.1082}&\textcolor{black}{0.0469}&\textcolor{black}{0.0796}&\textcolor{black}{0.0416}&\textcolor{black}{0.0558}&\textcolor{black}{0.0225}\\
    PCA&0.1283&0.0633&0.0748&0.0333&0.0625&0.0334&\textcolor{black}{0.0565}&\textcolor{black}{0.0227}\\
    BN&0.1628&0.0789&0.1150&0.0494&0.0799&0.0418&\textcolor{black}{\underline{0.0569}}&\textcolor{black}{0.0228}\\
    CD&\underline{0.1664}&\underline{0.0798}&\underline{0.1230}&\underline{0.0528}&\textbf{0.0891}&\textbf{0.0465}&\textcolor{black}{0.0565}&\textcolor{black}{\underline{0.0229}}\\
    ZCA&\textbf{0.1688}&\textbf{0.0810}&\textbf{0.1257}&\textbf{0.0537}&\underline{0.0888}&\underline{0.0462}&\textcolor{black}{\textbf{0.0586}}&\textcolor{black}{\textbf{0.0234}}\\

    \bottomrule
    \end{tabular}
    \vspace{-1em}
    \label{tab:whiten-methods}
\end{table*}

To examine the impact of different levels of decorrelation strength of whitening transformation on WhitenRec+, we experiment with two whitening transformations by fixing one of them with a group number $G$ of 1, representing fully whitened representations. We then vary the $G$ of relaxed whitened representations in $\{4, 8, 16, 32, 64, 128, \text{Raw}\}$. 
``Raw'' denotes text features without whitening transformation.
We also include WhitenRec's accuracy in the plots for comparison. 
The results, as depicted in Fig.~\ref{fig:group}, show consistent patterns across the three Amazon datasets. Initially, there is a trend where an increase in the number of groups $G$ corresponds to a decline in performance.
Furthermore, a larger value of $G$ results in inferior performance compared to WhitenRec. This indicates that overly relaxed whitened representations are not beneficial and could impede the model's performance. 
Therefore, it is advisable to opt for a smaller $G$ value that ensures a relatively stronger decorrelation strength when configuring $G$ for relaxed whitened representations.
\textcolor{black}{Specific to the Food dataset, the optimal number of groups is determined to be 64. Deviating from this number, whether by increasing or decreasing $G$, leads to a reduction in performance.}

\subsection{Effect of Projection Head}

We conduct a preliminary experiment to explore how the projection head affects the performance of WhitenRec+. 
We adjust the number of hidden layers within the projection head to the values of \{1,2,3\}, correspondingly denoted as MLP-1, MLP-2, and MLP-3. We also examine the model's performance when employing a linear projection devoid of a non-linear activation function, referred to as Linear. \textcolor{black}{Moreover, we implement the Mixture-of-Experts (MoE)~\cite{shazeer2017outrageously,hou2022towards} approach as an alternative projection head. }
Results are presented in Table~\ref{tab:head}. We observe that increasing the number of layers, and consequently the model's complexity, leads to improved performance. ``Linear'' performs the worst for all datasets except Toys, underscoring the importance of introducing a non-linear activation function to adapt the pre-trained text embeddings for downstream recommendation tasks. 
\textcolor{black}{The MoE projection head exhibits performance on par with MLP-1 across all datasets except Toys, where MoE is the least effective.}

\subsection{Whitening Transformations}
\label{sec:whiten}
We perform experiments to investigate the impact of utilizing different whitening transformations, including both non-parametric and parametric methods. 
The non-parametric methods examined are PCA, BN, CD, and ZCA. The parametric methods evaluated are PW~\cite{hou2022towards}, which employs a linear layer for whitening transformation, and \textcolor{black}{BERT-flow, which learns an invertible mapping to transform BERT embeddings into a latent Gaussian representation without loss of information~\cite{li2020emnlp}}.

From Table~\ref{tab:whiten-methods}, our results reveal that the parametric method PW generally performs inferiorly in comparison to the non-parametric methods. This can be attributed to the fact that a linear layer cannot ensure the transformed output is in fact whitened.
\textcolor{black}{BERT-flow can outperform PW and PCA across three Amazon datasets, demonstrating its capability in learning an effective invertible mapping.}
Among all the non-parametric whitening methods, PCA exhibits the worst performance due to the issue of stochastic axis swapping, which can impede training progress as noted in~\cite{huang2018decorrelated}. 
For three Amazon datasets, CD and ZCA outperform BN by producing more informative representations through further decorrelation between axes. CD and ZCA consistently rank as the best or second-best methods across all datasets. 
\textcolor{black}{Notably, for the Food dataset, all non-parametric whitening methods, except ZCA, demonstrate comparable performance. This may be attributed to the relatively short text descriptions (namely, recipe names) within this dataset, which likely contain minimal redundant information that would necessitate compression. Compared with Amazon datasets, which have an average word count of 20.5, the Food dataset averages only 3.8 words per description.}

\subsection{Effect of Ensemble Methods}
\textcolor{black}{We examine various ensemble techniques to integrate both fully and relaxed whitened item text representations. We evaluate element-wise summation (Sum), direct concatenation (Concat), and an attention mechanism (Attn). The results are presented in Table~\ref{tab:ensemble-method}. From our observations, for three Amazon datasets, Sum and Attn yield comparable results and both outperform Concat. In contrast, for the Food dataset, Sum, Concat, and Attn all exhibit comparable levels of performance.}

\begin{table}
    \caption{Performance comparison using different ensemble methods.}
    \centering
    \small
    \begin{tabular}{l@{\hspace{0.5\tabcolsep}}
    |c@{\hspace{0.7\tabcolsep}}c@{\hspace{0.7\tabcolsep}}c@{\hspace{0.7\tabcolsep}}|c@{\hspace{0.7\tabcolsep}}c@{\hspace{0.7\tabcolsep}}c}
     \toprule
    \multirow{2}{*}{Dataset}
    & \multicolumn{3}{c|}{R@20} & \multicolumn{3}{c}{N@20} \\   
    & Sum &Concat& Attn & Sum &Concat& Attn\\
    \midrule
    \multirow{1}{*}{Arts}&\textbf{0.1688}&0.1634&\underline{0.1640}&\textbf{0.0810}&0.0800&\underline{0.0803}\\
    \multirow{1}{*}{Toys}&\textbf{0.1257}&0.1187&\underline{0.1227}&\textbf{0.0537}&0.0515&\underline{0.0530}\\
    \multirow{1}{*}{Tools}&\underline{0.0888}&0.0854&\textbf{0.0892}&\underline{0.0462}&0.0445&\textbf{0.0465}\\
    \multirow{1}{*}{Food}&\textbf{0.0586}&\underline{0.0580}&\underline{0.0580}&\textbf{0.0234}&\textbf{0.0234}&\underline{0.0232}\\
    \bottomrule
    \end{tabular}
    \label{tab:ensemble-method}
    \vspace{-1em}
\end{table}

\subsection{Incorporating ID Embeddings}
\textcolor{black}{We examine the effects of utilizing both ID and text embeddings with WhitenRec and WhitenRec+. Following UniSRec~\cite{hou2022towards}, we adopt a straightforward element-wise summation to merge the text and ID embeddings, yielding the final item representation. The results are shown in Table~\ref{tab:with-ids}. Across all four datasets, the inclusion of ID embeddings negatively impacts performance. This could be due to two potential reasons. Firstly, the mere addition of ID embeddings may not represent the most effective method of integration with whitened text embeddings, necessitating further investigation into optimal strategies within the WhitenRec framework as future work. Secondly, it is hypothesized that while the integration of ID embeddings enhances user uniformity, it may have exceeded optimal levels, resulting in performance degradation.}

\begin{table}
    \renewcommand{\arraystretch}{0.9}
    \caption{\textcolor{black}{Performance comparison of WhitenRec and WhitenRec+ Using text or text+ID Embeddings.}}
    \centering
    \small
    \begin{tabular}{l c|cc|cccc}
     \toprule
    \multirow{2}{*}{Dataset} & \multirow{2}{*}{Metric} 
    & \multicolumn{2}{c|}{WhitenRec} & \multicolumn{2}{c}{WhitenRec+} \\   
    &&(T)&(T+ID)&(T)&(T+ID)\\
    \midrule
    \multirow{2}{*}{Arts} & R@20 &\textbf{0.1625}&0.1442&\textbf{0.1688}&0.1434\\
    & N@20 &\textbf{0.0796}&0.0786&\textbf{0.0810}&0.0787\\
    \midrule
    \multirow{2}{*}{Toys} & R@20&\textbf{0.1201}&0.1166&\textbf{0.1257}&0.1163 \\
    & N@20 &0.0521&\textbf{0.0532}&\textbf{0.0537}&0.0527\\
    \midrule
    \multirow{2}{*}{Tools} & R@20&\textbf{0.0861}&0.0756&\textbf{0.0888}&0.0741 \\
    & N@20 &\textbf{0.0453}&0.0449&\textbf{0.0462}&0.0453\\
    \midrule
    \multirow{2}{*}{Food} & R@20 &\textbf{0.0569}&0.0549&\textbf{0.0586}&0.0537\\
    & N@20 &\textbf{0.0227}&0.0222&\textbf{0.0234}&0.0220\\
    \bottomrule
    \end{tabular}
    \label{tab:with-ids}
    \vspace{-1em}
\end{table}

\begin{table}
  \centering
   \small
       \renewcommand{\arraystretch}{0.9}
   \caption{\textcolor{black}{Efficiency Comparison on Tools Dataset.}}
    \resizebox{\columnwidth}{!}{
  \begin{tabular}{c|cc|cc|ccccc}
    \toprule
    \multirow{2}{*}{Model} & \multicolumn{2}{c|}{UniSRec} &\multicolumn{2}{c|}{WhitenRec} & \multicolumn{2}{c}{WhitenRec+}&\\
    & (T) & (T+ID) & (T) & (T+ID) & (T) & (T+ID)\\
    \midrule
    \#Params & 2.9M &13.8M&1.4M&12.2M&1.4M&12.2M\\
    s/Epoch &90 &99 &63 & 75 &64 & 77\\
    \bottomrule
  \end{tabular}}
  \vspace{-1.2em}
  \label{tab:efficiency}
\end{table}

\subsection{Efficiency Analysis}
\textcolor{black}{Our comparative analysis of WhitenRec and WhitenRec+ with the leading baseline, namely UniSRec, focuses on the parameter size and training time per epoch, as detailed in Table~\ref{tab:efficiency}.
First, the integration of ID embeddings substantially increases the parameter count, resulting in a training time that is approximately 10\% longer compared to models without ID embeddings. This pattern is consistent across all these three methods and suggests a potential compromise between model complexity and efficiency due to the addition of item ID embeddings.
Second, WhitenRec and WhitenRec+ leverage pre-computable whitening transformations to enhance performance without adding complexity. By relying solely on text embeddings, our models maintain a lower number of parameters, which not only mitigates overfitting but also confers benefits in situations of cold-start, as shown in Table~\ref{tab:cold}.}

\section{Conclusion}
In conclusion, we present the frameworks WhitenRec and WhitenRec+ to effectively exploit text features of items in sequential recommendation. 
We contend that relying on text embeddings from pre-trained language models is sub-optimal because such embeddings exist in an anisotropic semantic space, which limits the differentiation among item representations.
To address this issue, we propose the WhitenRec method, which transforms the anisotropic text embedding distribution into an isotropic distribution through whitening. 
When an excessive whitening transformation is applied, text embeddings can deviate from their original semantics. 
Relying solely on relaxed whitening results in a clustered embedding distribution and sub-optimal performance.
To benefit from both ends, we introduce WhitenRec+, which leverages both fully whitened and relaxed whitened item representations to balance differentiation and similarity. Our experimental results on four benchmark datasets demonstrate that our proposed methods outperform existing state-of-the-art models for sequential recommendation on both warm and cold settings.

\section*{Acknowledgment}

This research is supported by Alibaba Group and Alibaba-NTU Singapore Joint Research Institute(JRI), Nanyang Technological University, Singapore.

\bibliographystyle{IEEEtran}
\bibliography{ref}

\end{document}